\documentclass{article}

\usepackage{microtype}
\usepackage{graphicx}
\usepackage{subfigure}
\usepackage{wrapfig}
\usepackage{tikz}
\usetikzlibrary{positioning, arrows.meta}
\usepackage{tikz-cd}
\usepackage{booktabs} 

\usepackage{hyperref}

\usepackage[accepted]{icml2024}

\usepackage{amsmath}
\usepackage{amssymb}
\usepackage{mathtools}
\usepackage{amsthm}

\usepackage[capitalize,noabbrev]{cleveref}

    \theoremstyle{plain}
    \newtheorem{theorem}{Theorem}[section]
    \newtheorem{proposition}{Proposition}[section]
    \newtheorem{corollary}{Corollary}[section]
    \newtheorem{lemma}{Lemma}[section]
    \theoremstyle{definition}
    \newtheorem{definition}[theorem]{Definition}
    
    \theoremstyle{remark}
    \newtheorem{remark}[theorem]{Remark}
    \usepackage{stmaryrd}
    \usepackage{xparse} \DeclarePairedDelimiterX{\Iintv}[1]{\llbracket}{\rrbracket}{\iintvargs{#1}}
    \NewDocumentCommand{\iintvargs}{>{\SplitArgument{1}{,}}m}
    {\iintvargsaux#1} %
    \NewDocumentCommand{\iintvargsaux}{mm} {#1\mkern1.5mu,\mkern1.5mu#2}
    
    \DeclareMathOperator{\Lap}{Lap}
    \DeclareMathOperator{\GCauchy}{GCauchy}
    \DeclareMathOperator{\Cauchy}{Cauchy}
    
    \DeclareMathOperator{\Cov}{Cov}
    \DeclareMathOperator{\Var}{Var}
    \newcommand{\CNI}{\mathit{CNI}}
    
    \newtheoremstyle{restated}
      {}
      {}
      {\itshape}
      {}
      {\bfseries}
      {.}
      {.5em}
      {\thmname{#1}\thmnumber{ \ref{#3}}\thmnote{ (#3)}}
    
    \theoremstyle{restated}

\usepackage[textsize=tiny]{todonotes}

\icmltitlerunning{General Additive Noise Mechanisms and Privacy Amplification by Iteration for Pufferfish Privacy}

\begin{document}

\twocolumn[
\icmltitle{Rényi Pufferfish Privacy: General Additive Noise Mechanisms and Privacy Amplification by Iteration via Shift Reduction Lemmas}




\begin{icmlauthorlist}
\icmlauthor{Clément Pierquin}{Craft,Lille}
\icmlauthor{Aurélien Bellet}{Montpellier}
\icmlauthor{Marc Tommasi}{Lille}
\icmlauthor{Matthieu Boussard}{Craft}

\end{icmlauthorlist}

\icmlaffiliation{Craft}{Craft AI, Paris, France}
\icmlaffiliation{Lille}{Université de Lille, Inria, CNRS, Centrale Lille, UMR 9189 CRIStAL, F-59000 Lille, France}
\icmlaffiliation{Montpellier}{
Inria, Université de Montpellier}
\icmlcorrespondingauthor{Clément Pierquin}{clement.pierquin@craft-ai.fr}

\icmlkeywords{Differential privacy, Pufferfish privacy}

\vskip 0.3in
]



\printAffiliationsAndNotice{}  

\begin{abstract}
  Pufferfish privacy is a flexible generalization of differential privacy that allows to model arbitrary secrets and adversary's prior knowledge about the data.
  Unfortunately, designing general and tractable Pufferfish mechanisms that do not compromise utility is challenging. Furthermore, this framework does not provide the composition guarantees needed for a direct use in iterative machine learning algorithms. To mitigate these issues, we introduce a Rényi divergence-based variant of Pufferfish and show that it allows us to extend the applicability of the Pufferfish framework. We first generalize the Wasserstein mechanism to cover a wide range of noise distributions and introduce several ways to improve its utility. 
  Finally, as an alternative to composition, we prove privacy amplification results for contractive noisy iterations and showcase the first use of Pufferfish in private convex optimization. A common ingredient underlying our results is the use and extension of shift reduction lemmas.
\end{abstract}

\section{Introduction}

Differential privacy (DP)~\citep{Dwork2014} is now considered as the gold standard for privacy-preserving data analysis.
However, despite its many desirable properties, DP does not suit all types of data effectively.
Specifically, the guarantees it offers are based on the underlying assumption that individuals in the dataset being analyzed are statistically independent. In reality, data often exhibit correlations, and when two correlated individuals are present in a dataset, performing the same analysis with and without one of these individuals could leak more knowledge about the individual than the conventional differential privacy framework assumes~\citep{Humphries2023}.

\looseness=-1
To address these situations, specialized privacy definitions have been designed. Certain direct extensions of DP, like group privacy~\citep{Dwork2014} or entry privacy~\citep{Hardt2013}, protect entire instances or groups, which results in strong privacy guarantees but often much poorer utility. More flexible frameworks allow to tailor the privacy definition to a set of distributions which could have plausibly generated the dataset, and thereby allow a tighter privacy analysis.
In this work, we focus on the general framework of Pufferfish privacy~\citep{Kifer2014}, which is closely related to other similar definitions like Blowfish privacy~\citep{He2014} and distribution privacy~\citep{Kawamoto2019,Chen2023}.

Pufferfish privacy however comes with new challenges, first and foremost in the design of general and computationally tractable Pufferfish private mechanisms. Indeed, the sensitivity of the query, which is critical in DP to design additive noise mechanisms,
has no direct use in Pufferfish privacy. 
Moreover, while various ways to measure and efficiently track the privacy loss have been proposed for DP, see for instance Rényi differential privacy (RDP)~\citep{Mironov2017}, this flexibility is lacking in Pufferfish privacy.
As a result, previous work on the design of Pufferfish mechanisms has focused on specific noise distributions and applications~\citep{Kifer2014,Ou2018,Kessler2015,Niu2019,Song2017}.
For instance,~\citet{Song2017} proposed the Wasserstein mechanism for the Laplace noise, which relies on the computation of $\infty$-Wasserstein distances. 
Another recent work proposes an exponential mechanism-based approach which provides a more computationally tractable approach but relies on (potentially loose) sufficient conditions for Pufferfish privacy~\citep{Ding2022}.
The Pufferfish framework thus lacks a unified theory that subsumes the original worst-case definition and allows for the design of general additive mechanisms compatible with a wide range of noise distributions.

\looseness=-1 Another key limitation of Pufferfish privacy is that it does not always compose when the same data is used across multiple computations. Existing sequential and adaptive compositions theorems hold only for some particular Pufferfish instantiations and mechanisms, sometimes without a closed-form that can be used in practice~\citep{Kifer2014,DBLP:journals/tit/NuradhaG23}. For instance, a sequential (but non-adaptive) composition result exists for the Markov Quilt Mechanism \citep{Song2017}, but it is limited to Bayesian networks. The lack of a universal adaptive composition theorem currently makes Pufferfish privacy unfit for the analysis of iterative algorithms such as those used in differentially private machine learning \citep{Abadi2016}.




In this paper, we mitigate the above limitations of Pufferfish privacy by making the following  contributions:
\begin{itemize}
    \item We define the Rényi Pufferfish privacy framework and show that it preserves the main desirable properties of Pufferfish while providing additional flexibility.
    \item We introduce the General Wasserstein Mechanism (GWM), a generalization of the Wasserstein mechanism of~\citet{Song2017}. Our mechanism allows to derive (Rényi) Pufferfish privacy guarantees for all additive noise distributions that are absolutely continuous with respect to the Lebesgue measure. 
    \item We propose two ways to improve the utility of GWM by relaxing the $\infty$-Wasserstein distance used to calibrate the noise. Our first approach relies on a $\delta$-approximation allowing the tail of the distribution of the mechanism to be disregarded, similar to what has been proposed by~\citet{Chen2023} for the distribution privacy framework. Incidentally, we demonstrate an equivalence between Pufferfish privacy and distribution privacy. Our second approach enables the use of $p$-Wasserstein distances, yielding the first general Pufferfish mechanism with better utility than the Wasserstein mechanism at the same privacy cost.
    \item Inspired by~\citet{Feldman2018}, we prove privacy amplification by iteration results for Pufferfish, allowing to bypass the use of composition in the analysis of contractive noisy iterations. This technique is particularly useful to analyze convex optimization with stochastic gradient descent, allowing the integration of Pufferfish privacy in machine learning pipelines.
    \item We provide examples of concrete instantiations of our framework where the proposed mechanisms are computationally efficient and provide better utility than (Group) DP.
\end{itemize}

One of our key technical contributions lies in the novel use and generalization of shift reduction lemmas~\citep{Feldman2018,Altschuler2022} in the context of Pufferfish privacy. We argue that shift reduction is the right tool to analyze Pufferfish privacy, and believe this view may yield more results in the future.

All proofs and some additional content can be found in the supplementary material.

\section{Rényi Pufferfish Privacy}
\label{sectionrenyi}

Rényi differential privacy (RDP) ensures that an adversary cannot gain too much knowledge about whether an individual point is in the dataset or not by observing the output of the mechanism. In the original definition, it is implied that the elements of the dataset are statistically independent (see Appendix~\ref{AppendixRDPRefs} for definitions). A more general framework, Pufferfish privacy, has been designed to handle possibly correlated data and other types of secrets than the presence of an individual in a dataset~\citep{Kifer2014}. In a Pufferfish instantiation, we denote by $\mathcal{S}$ the set of possible secrets to be protected, and by $\mathcal{Q} \subseteq \mathcal{S}^2$ the specific pairs of secrets we aim to make indistinguishable. In contrast to differential privacy, the variable $X$ representing the dataset is not deterministic in Pufferfish privacy. Instead, it is sampled from a certain distribution $\theta \in \Theta$. The set $\Theta$ represents the possible prior knowledge of an adversary.

\begin{definition}[Pufferfish privacy, PP \citep{Kifer2014,Ding2022}]
Let $\varepsilon \geq 0$ and $\delta \in (0,1)$. A privacy mechanism $\mathcal{M}$ is said to be $(\varepsilon,\delta)$-Pufferfish private in a framework $(\mathcal{S}, \mathcal{Q}, \Theta)$ if for all $\theta \in$ $\Theta$, for all secret pairs $\left(s_{i}, s_{j}\right) \in \mathcal{Q}$, and for all $w \in \operatorname{Range}(\mathcal{M})$, we have:
\[P\left(\mathcal{M}(X)=w \mid s_{i}, \theta\right)\leq e^{\varepsilon}P\left(\mathcal{M}(X)=w \mid s_{j}, \theta\right) + \delta,\]
where $X\sim\theta$ and $(s_{i},s_{j})$ is such that $P\left(s_{i} \mid \theta\right) \neq 0, P\left(s_{j} \mid \theta\right) \neq 0$. If $\delta = 0$, $\mathcal{M}$ satisfies $\varepsilon$-Pufferfish privacy.
\end{definition}

In this work, we introduce a Rényi divergence-based version of Pufferfish privacy. Using Rényi divergences in privacy definitions 
has several advantages. Especially relevant to our work will be the quantification of privacy guarantees by
bounding certain moments of the exponential of the privacy loss~\citep{Mironov2017}, and the ability to leverage a large body of results on Rényi divergences such as shift reduction lemmas~\citep{Feldman2018,Altschuler2023}.

\begin{definition}[Rényi Pufferfish privacy, RPP]
Let $\alpha > 1$ and $\varepsilon \geq 0$. A privacy mechanism $\mathcal{M}$ is said to be $(\alpha, \varepsilon)$-Rényi Pufferfish private in a framework $(\mathcal{S}, \mathcal{Q}, \Theta)$ if for all $\theta \in$ $\Theta$ and for all secret pairs $\left(s_{i}, s_{j}\right) \in \mathcal{Q}$, we have:
\[D_\alpha\left(P\left(\mathcal{M}(X) \mid s_{i}, \theta\right), P\left(\mathcal{M}(X) \mid s_{j}, \theta\right)\right) \leq \varepsilon,\]
where $X\sim\theta$, $(s_{i},s_{j})$ is such that $P\left(s_{i} \mid \theta\right) \neq 0$ and $P\left(s_{j} \mid \theta\right) \neq 0$, and $D_\alpha(\mu,\nu) = \frac{1}{\alpha-1}\log\mathbb{E}_{x \sim \nu}\left[\left(\frac{\mu(x)}{\nu(x)}\right)^\alpha\right]$ is the Rényi divergence of order $\alpha$.
\end{definition}

Rényi Pufferfish privacy upholds the post-processing inequality, which is a key attribute for any effective privacy framework.

\begin{proposition}[Post-processing]
    \label{postproc}
    Let $\mathcal{M}_1$ be a randomized algorithm and $\mathcal{M}$ be $(\alpha,\varepsilon)$-RPP. Then,
    \begin{align*}
        &D_\alpha\left(P\left(\mathcal{M}_1(\mathcal{M}(X)) \mid s_{i}, \theta\right), P\left(\mathcal{M}_1(\mathcal{M}(X)) \mid s_{j}, \theta\right)\right) \\ &\leq D_\alpha\left(P\left(\mathcal{M}(X) \mid s_{i}, \theta\right), P\left(\mathcal{M}(X) \mid s_{j}, \theta\right)\right) \leq \varepsilon.
    \end{align*}
\end{proposition}

It is easy to see that $(\infty,\varepsilon)$-RPP corresponds to $\varepsilon$-PP. Furthermore, $(\alpha,\varepsilon)$-RPP can be converted to $(\varepsilon,\delta)$-PP.

\begin{proposition}[RPP implies PP]
\label{RPPtoPP}
If $\mathcal{M}$ is $(\alpha,\varepsilon)$-RPP, it also satisfies $\big(\varepsilon + \frac{\log(1/\delta)}{\alpha-1},\delta\big)$-PP  $\forall\delta \in (0,1)$.
\end{proposition}

\paragraph{Guarantees against close adversaries.} In Pufferfish, the set $\Theta$ represents the possible beliefs of the adversary. It needs to be large enough to prevent harmful privacy leaks, but there is also a no free lunch theorem that states that if $\Theta$ is too large then the resulting mechanism will have poor utility~\cite{Kifer2014}. Hence, it is important to quantify the privacy protection offered by a mechanism $\mathcal{M}$ when the belief $\theta'$ of the adversary is not in $\Theta$. This question has been addressed for $\epsilon$-PP by~\citet{Song2017}.
The theorem derived by~\citet{Song2017}, which we recall in Appendix~\ref{closeappendix} for completeness, shows that if $\theta'$ is $\Delta$-close to some $\theta\in\Theta$, then $\mathcal{M}$ retains its Pufferfish privacy guarantees for $\theta'$ up to an additive penalty $2\Delta$. However, $\Delta$ is measured in $\infty$-Rényi divergence, which corresponds to a worst-case scenario, and can thus be very large. We extend this result to our RPP framework, allowing the use of $\alpha$-Rényi divergences (see Appendix~\ref{closeappendix}). Our result can provide better privacy guarantees in situations where the original one gives poor guarantees. 


\paragraph{Running examples.}
\label{example}

\looseness=-1 We introduce here some examples of RPP instantiations which we will use throughout the paper to illustrate our private mechanisms. Let $n>0$ be the total number of participants in a study. Let $\mathcal{X}$ be the potential values of an individual's private features. Let $X=  (X_{1},\dots,X_{n})\in\mathcal{X}^n$ describing the private properties of the $n$ individuals. An adversary anticipates correlations among individuals within the study with a prior $\theta \in \Theta$. We define the set of secrets for this adversary as $\mathcal{S} = \{s_{i}^a \triangleq \{X_{i} = a\} ;\; a \in \mathcal{X}, i \in \Iintv{1,n}\}$ 
and define $\mathcal{Q} = \{(s_{i}^a,s_{j}^b);\; a,b \in \mathcal{X}, i,j \in \Iintv{1,n}\}$. Consider the following simple instantiations of this setting for datasets of size 2:
\begin{trivlist}\itshape

 \item \emph{\textbf{\hypertarget{example1}{Example 1}}} (Counting query with correlation). Each individual $i$ holds a binary value $X_{i} \in \{0,1\}$ and we consider a counting query $f(X) = X_{1}+X_2$. For $p \in (0,1), \rho \in [-1,1]$, the adversary has the following prior: $P(X_1 = 1) = P(X_2 = 1) = p$, where $X_1$ and $X_2$ are drawn with correlation $\rho$.
    \item \emph{\textbf{\hypertarget{example2}{Example 2}}} (Average salary query). Each individual $i$ holds her salary $X_{i} \geq 0$ and we consider an average query $f(X) = \frac{1}{2}(X_1+X_2)$. The adversary has the following prior for the marginals: for $i \in \{1,2\}$, $$X_{i} =  \begin{cases}
      1 & \text{with prob. $1/2$}\\
      2 & \text{with prob. $499/1000$,}\\
      100 & \text{with prob. $1/1000$}
    \end{cases}~\text{ for }i \in \{1,2\}$$
    Here, $X_1$ and $X_2$ are thus considered independent.
    \item \emph{\textbf{\hypertarget{example3}{Example 3}}} (Sum query with user-dependent prior). We consider $\mathcal{X} = (0,r)$ and a sum query $f(X) = X_1+X_2$. The adversary has an arbitrary prior about the distribution of $(X_1,X_2)$ but assumes that each individual $i$ holds a different value $X_{i} \in (0,r_i)$ with $0 < r_i \leq r$. 

    \end{trivlist}

\section{A General Additive Mechanism for Rényi Pufferfish Privacy}
\label{sectionwassersteinmechanism}
In this section, we present a general approach to obtain Rényi Pufferfish privacy guarantees. Specifically, we introduce the General Wasserstein Mechanism (GWM), a generalization of the Laplacian-based Wasserstein mechanism of~\citet{Song2017} to a wide range of noise distributions, and derive the corresponding RPP guarantees. We also highlight that the shift reduction lemma and its variants, introduced by~\citet{Feldman2018} in the context of privacy amplification by iteration, provide the right framework for analyzing Rényi Pufferfish privacy.

We first introduce $\infty$-Wasserstein distances and couplings.

\begin{definition}[Couplings]
    Let $\mu$ and $\nu$ be two distributions on a measurable space $(\mathbb{R}^d, \mathcal{B}(\mathbb{R}^d))$ with $\mathcal{B}(\mathbb{R}^d))$ the Borel $\sigma$-algebra. A coupling $\pi$ is a joint distribution on the product space $(\mathbb{R}^{d\times2}, \mathcal{B}(\mathbb{R}^d)^2)$ with marginals $\mu$ and $\nu$, where $\mathcal{B}(\mathbb{R}^d)^2$ is the product $\sigma$-algebra.
\end{definition}
\begin{definition}[$\infty$-Wasserstein distance]
    Let $\mu$ and $\nu$ be two distributions on $\mathbb{R}^d$. We note $\Gamma$ the set of the couplings between $\mu$ and $\nu$. We define the $\infty$-Wasserstein distance between $\mu$ and $\nu$ as:
    \[W_\infty(\mu,\nu) = \underset{\pi \in \Gamma(\mu,\nu)}{\inf}\underset{(x,y) \in supp(\pi)}{\sup}\|x-y\|.\]
    Throughout the paper, $\|\cdot\|$ represents a norm of $\mathbb{R}^d$. When necessary, in later results, the type of norm will be specified.
\end{definition}

We now recall the shift reduction lemma, a result that allows to split the Rényi divergence between two noised distributions into two distinct components: one involving the two original distributions, and one involving the noise.
Let $\mu, \nu, \zeta$ be three distributions on $\mathbb{R}^d$ and $z,a \geq 0$.
We define the following quantities:
\[D_\alpha^{(z)}(\mu,\nu) = \inf_{W_\infty(\mu,\mu')\leq z}D_\alpha(\mu',\nu),\] \[R_\alpha(\zeta,z) = \sup_{\|x\|<z}D_\alpha(\zeta_{-x},\zeta),\] where $\zeta_{-x} : y \mapsto \zeta(y-x)$, and denote by $\ast$ the convolution product.
\begin{lemma}[Shift reduction \citep{Feldman2018}]
\label{shiftreduction}
Let $\mu, \nu, \zeta$ be three distributions on $\mathbb{R}^d$ and $z,a \geq 0$. Then,
\[D_\alpha^{(a)}(\mu \ast \zeta,\nu \ast \zeta) \leq D_\alpha^{(z+a)}(\mu,\nu) + R_\alpha(\zeta,z).\]
\end{lemma}

We now show that the shift reduction lemma allows to obtain a unified approach for RPP analysis. In fact, it gives a closed formula for the privacy guarantees of releasing a query with additive noise. This yields our General Wasserstein Mechanism (GWM) and its associated privacy guarantees.

\begin{theorem}[General Wasserstein mechanism, GWM]
\label{GWM}
Let $f : \mathcal{D} \to \mathbb{R}^d$ be a numerical query and denote:

\resizebox{\hsize}{!}{$\Delta_G = \underset{\substack{(s_i,s_j) \in S\\ \theta\in \Theta}}{\max} W_\infty \left( P(f(X)|s_i,\theta),P(f(X)|s_j,\theta)\right)$.}

Let $N = (N_1,\dots,N_d) \sim \zeta$ drawn independently of the dataset $X$. Then, $\mathcal{M}(X) = f(X) + N$ satisfies $(\alpha,R_\alpha(\zeta,\Delta_G))$-RPP for all $\alpha \in (1,+\infty)$ and $R_\infty(\zeta,\Delta_G)$-PP. 
\end{theorem}

While Theorem~\ref{GWM} is very general, we can easily derive explicit results for specific choices of noise distributions. Instantiating GWM with Laplacian noise, we recover the results of~\citet{Song2017} for PP as a special case where $d=1$. More interestingly, we also directly obtain a novel Gaussian mechanism and a novel Laplacian mechanism for RPP.


\begin{corollary}[Privacy guarantees for usual noise distributions]
\label{wassersteinusualdistrib}
We note $I_d$ the identity matrix of size $d$. Plugging the expressions of $R_\infty(\zeta,z)$ and $R_\alpha(\zeta,z)$ for Laplacian and Gaussian distributions, we obtain:
\begin{itemize}
    \item $\mathcal{M}(X) = f(X) + N$ with $N \sim \mathcal{N}\big(0,\frac{\alpha \Delta_G^2}{2\varepsilon}I_d\big)$ and $\Delta_G$ computed w.r.t. the $l_2$ norm is $(\alpha, \varepsilon)$-RPP.
    \item $\mathcal{M}(X) = f(X) + L$ with $L \sim \Lap(0,\rho I_d)$ and $\Delta_G$ computed w.r.t. the $l_1$ norm is $\Big(\alpha,\frac{1}{\alpha-1}\log\big(\frac{\alpha}{2\alpha-1}e^{\Delta_G(\alpha-1)/\rho} + \frac{\alpha-1}{2\alpha-1}e^{-\Delta_G \alpha/ \rho} \big)\Big)$-RPP.
    \item $\mathcal{M}(X) = f(X) + L$ with $L \sim \Lap\left(0,\frac{\Delta_G}{\varepsilon}I_d\right)$ with $\Delta_G$ computed w.r.t. the $l_1$ norm is $\varepsilon$-PP.
\end{itemize}
\end{corollary}

The results of Corollary~\ref{wassersteinusualdistrib} are analogous to the results of~\citet{Mironov2017} for RDP, where the sensitivity of the query is replaced by $\Delta_G$. It enables us to directly compare the utility of a RDP mechanism in the group privacy setting and the GWM in RPP. Considering~\hyperref[example]{Example 3}, we have $\Delta_G \leq r_1 + r_2$, which is smaller than $\Delta_{\text{GROUP}} = 2r$. Therefore, GWM achieves better utility than group RDP in this case. This observation can be generalized to other settings as the utility guarantees of the Wasserstein mechanism of~\citet{Song2017} extend to the GWM.

\begin{proposition}[Utility of the GWM, informal]
\label{GWMUtility}
    Under mild conditions, an additive mechanism offers better utility in the GWM setting than in the group privacy setting (see Appendix~\ref{appendixGWMUtility} for details).
\end{proposition}
One drawback of GWM is that in some cases, $\Delta_G$ may be large, as it depends on $\infty$-Wasserstein distances. In~\hyperref[example]{Example 1}, $\Delta_G = \Delta_{\text{GROUP}} = 2$, thus GWM gives no utility advantage compared to group RDP. In~\hyperref[example]{Example 2}, $\Delta_G = 98$ is large although the event $X_i = 100$ is rare. We deal with this issue in the next section.

\section{Improving Utility by Relaxing the $W_\infty$ Constraint}

In this section, we propose two ways to improve the utility of GWM by relaxing the $\infty$-Wasserstein constraint in the calibration of the noise.
Figure~\ref{fig:graphmechanisms} summarizes the relations between the different mechanisms and privacy definitions that we introduce.
\begin{figure}
\label{fig:graphmechanisms}
\begin{tikzcd}[column sep=small, row sep=normal]
 R_\alpha(\zeta,\Delta_{G,\delta}) \arrow[r, dashrightarrow, "\leq"]
 \arrow[dd, bend right=60] 
 & R_\alpha(\zeta,\Delta_{G})  \arrow[r, dashrightarrow, "\leq"] \arrow[dd]
 & R_\alpha(\zeta,\Delta_{\text{GROUP}})\arrow[dd]\\
 \frac{\log\Delta_G^{\zeta,1,\alpha}}{\alpha-1} \arrow[ru, dashrightarrow, "\leq"]
 \arrow[rd]  & & \\
 (\delta,\alpha,\varepsilon)\text{-RPP} \arrow[r, "\underset{\delta \to 0}{\implies}" description]\arrow[r, Leftarrow, bend right] & (\alpha,\varepsilon)\text{-RPP} \arrow[r, Leftarrow] & (\alpha,\varepsilon)\text{-GROUP-RDP} 
\end{tikzcd}
\caption{Relations between the mechanisms and privacy notions studied in the paper. The values on the top of the graph represent the value $\varepsilon$ of the privacy budget guaranteed by the mechanisms. $\Delta_{G}$ corresponds to the sensitivity of the GWM (Section~\ref{sectionwassersteinmechanism}), and $\Delta_{G,\delta}$ corresponds to the sensitivity of the GAWM (Section~\ref{sectionapproxwasserstein}),  $\Delta_G^{\zeta,1,\alpha}$ corresponds to the sensitivity of the DAGWM (Section~\ref{sectionrelaxWasserstein}). $\Delta_{\text{GROUP}}$ corresponds to the sensitivity of mechanisms in the group privacy framework. The plain arrows indicate the privacy guarantees offered by the mechanisms. The dashed arrows compare the privacy budget offered by the mechanisms. The implication arrows illustrate the relations between the different frameworks.}
\end{figure}
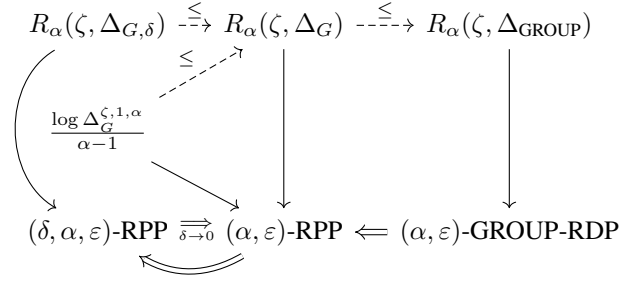

\subsection{$\delta$-Approximation of $(\alpha,\varepsilon)$-RPP}
\label{sectionapproxwasserstein}

Our first approach is to define an approximation of Rényi Pufferfish Privacy that allows a low probability set of values to be disregarded. 
\begin{definition}[Approximate Rényi Pufferfish privacy]
A privacy mechanism $\mathcal{M}$ is said to be $(\alpha, \varepsilon, \delta)$-approximate Rényi Pufferfish private in a framework $(\mathcal{S}, \mathcal{Q}, \Theta)$ if for all $\theta \in$ $\Theta$ and for all secret pairs $\left(s_{i}, s_{j}\right) \in \mathcal{Q}$, there exists $E, E'$ such that $P(E) \geq 1- \delta, P(E') \geq 1 - \delta$ and:
\begin{align*}
D_\alpha\left(P\left(\mathcal{M}(X) \mid s_{i}, \theta, E\right), P\left(\mathcal{M}(X) \mid s_{j}, \theta, E'\right)\right) &\leq \varepsilon, \\ D_\alpha\left(P\left(\mathcal{M}(X) \mid s_{j}, \theta, E'\right), P\left(\mathcal{M}(X) \mid s_{i}, \theta, E\right)\right) &\leq \varepsilon,
\end{align*}
where $X\sim\theta$ and $(s_{i},s_{j})$ is such that $P\left(s_{i} \mid \theta\right) \neq 0, P\left(s_{j} \mid \theta\right) \neq 0$.
\end{definition}
Note that similar privacy definitions have been proposed for versions of differential privacy in~\citep[][Definition 8.1]{Bun2016} and~\citep[][Definition 18]{Papernot2021}. This definition implies $(\varepsilon,\delta)$-PP when $\alpha \to + \infty$. It also implies $(\varepsilon',2\delta)$-RPP for a specific value $\varepsilon'$.
\begin{proposition}
\label{aproxRPPtoRPP}
     If $\mathcal{M}$ is $(\alpha,\varepsilon,\delta)$-approximate RPP, then it is $(\varepsilon',2\delta)$-PP, with $\varepsilon' = \varepsilon + \frac{\log(1/\delta)}{\alpha-1}$.
 \end{proposition}
 

We now design an approximate Wasserstein mechanism for Rényi Pufferfish privacy. To do so, we rely on the notion of $(z,\delta)$-proximity~\citep[named \textit{closeness} in][]{Chen2023}.
\begin{definition}[$(z,\delta)$-proximity]
\label{def:close}
    Let $\mu, \nu$ two distributions on $\mathbb{R}^d$ and $z \geq 0, \delta \in (0,1)$. We say that $\mu$ and $\nu$ are $(z,\delta)$-near if there exists a coupling $\pi$ between $\mu$ and $\nu$ and $\mathcal{R}\subset supp(\pi)$ such that $\int_\mathcal{R}d\pi(x,y) \geq 1- \delta$ and $\forall (x,y) \in \mathcal{R}, \|x-y\| \leq z$.
\end{definition}

We also need to extend the shift reduction lemma of~\citet{Feldman2018} to account for shifts that are $(z,\delta)$-near to the original distribution $\mu$, instead of shifts $\mu'$ such that $W_\infty(\mu,\mu') \leq z$. 
\begin{lemma}[Approximate shift reduction]
\label{shiftApprox}
Let $\mu, \nu, \zeta$ be three distributions on $\mathbb{R}^d$. We denote $D_\alpha^{(z, \delta)}(\mu,\nu) = \underset{\mu,\mu'~(z,\delta)\text{-near}}{\inf}D_\alpha(\mu',\nu)$. Then, for all $\delta \in (0,1)$, there exists an event $E$ such that $P(E) \geq 1-\delta$ and:  
\begin{align*}
    &D_\alpha\left((\mu \ast \zeta)_{|E},(\nu \ast \zeta)\right) \\&\leq D_\alpha^{(z,\delta)}(\mu,\nu) + R_\alpha(\zeta,z) + \frac{\alpha}{\alpha-1}\log\Big(\frac{1}{1-\delta}\Big).
\end{align*}
\end{lemma}
This approximate shift reduction lemma provides a general mechanism to achieve approximate RPP. 
\begin{theorem}[General approximate Wasserstein mechanism, GAWM]
\label{GAWM}
Let $f : \mathcal{D} \to \mathbb{R}^d$ be a numerical query. 
 For all $\delta \in (0,1)$, let us denote:
\begin{align*}
    &\Delta_{G,\delta} >\inf\{ z \in \mathbb{R}; \forall (s_i,s_j) \in S, \forall \theta \in \Theta, \\&\left( P((f(X)|s_i,\theta),P(f(X)|s_j,\theta)\right) \text{ are } (z,\delta) \text{-near}\}.
\end{align*}
Let $N = (N_1,\dots,N_d) \sim \zeta$ drawn independently of the dataset $X$. Then, $\mathcal{M} = f(X) + N$ satisfies $(\alpha,R_\alpha(\zeta,\Delta_{G,\delta})+ \frac{\alpha}{\alpha-1}\log\frac{1}{1-\delta},\delta)$-approximate RPP for all $\alpha \in (1,+\infty)$ and $(R_\infty(\zeta,\Delta_{G,\delta}) + \log\frac{1}{1-\delta},\delta)$-PP. 
\end{theorem}
From this general result, we can then design approximate RPP mechanisms for usual noise distributions. These results are similar to those of the general Wasserstein mechanism (see~Corollary~\ref{wassersteinusualdistrib}) but with an additive term that depends on $\delta$. We refer to Appendix~\ref{completeapproxwasserstein} for details. Using this new mechanism, we can obtain better utility at a small privacy cost for queries that take large values with small probability. In~\hyperref[example]{Example 2}, we have $\Delta_G = 98$ while for $\delta=3\cdot 10^{-3}$, $\Delta_{G,\delta} = 1$, which yields a major improvement in utility.
This observation also holds in a more general case.
\begin{proposition}[Utility of the GAWM, informal]
\label{GAWMUtility}
    At a privacy cost of $\delta \in (0,1)$, the GAWM offers more utility than the GWM (see Appendix~\ref{appendixGAWMUtility} for details).
\end{proposition}

\begin{remark}[Relation to distribution privacy]
\label{distribRelat}
A related result has been shown 
by~\citet{Chen2023} for the distribution privacy framework (see Appendix~\ref{proofDistrPP} for the definition of distribution privacy and the result). The formulation of the results are similar, despite employing a different proof technique to get the conclusions. We prove a connection between the two results by establishing a formal equivalence between Pufferfish privacy and distribution privacy, which appears to be novel and could be of independent interest. In the interest of space, we refer to Appendix~\ref{proofDistrPP} for the formal result and its proof. While our approximate shift reduction result (Lemma~\ref{shiftApprox}) induces an additional term which prevents us from recovering exactly the results of~\citet{Chen2023} in the particular case of the Laplace mechanism for PP, our result can be used with a wide range of noise distributions and in the RPP framework, which is more general than PP (and thus more general than distribution privacy).
\end{remark}


\subsection{Leveraging $p$-Wasserstein Metrics}
\label{sectionrelaxWasserstein}

As another way to improve the utility of the GWM, we propose to use shifts constrained by $p$-Wasserstein metrics instead of $\infty$-Wasserstein metrics, thereby replacing the worst case transportation cost between $P(f(X)|s_i,\theta)$ and $P(f(X)|s_j,\theta)$ by moments of the transportation cost. This idea was explored in a different context by~\citet{Altschuler2023}, who considered Orlicz-Wasserstein shifts for Gaussian noise and identified a dependency between the noise distribution and the selected Wasserstein shift constraint. They argue that the Orlicz-Wasserstein metric is the ``right'' metric to use for the shifted Rényi analysis because the original shift reduction lemma fails for weaker shifts. Inspired by these considerations, we broaden the applicability of the Orlicz-Wasserstein shift reduction lemma of~\citet{Altschuler2023} by adapting their result to a wider range of noise distributions.

\begin{lemma}[Generalized shift reduction]
\label{generalizedShift}
    Let $\zeta$ be a noise distribution of $\mathbb{R}^d$. Let $z, p ,q > 0$ such that $1/p+1/q=1$.     We note:
    \resizebox{\hsize}{!}{$D_{\alpha, \alpha', \zeta}^{(z)}(\mu,\nu) = \underset{\xi; \underset{W \sim \xi}{\mathop{\mathbb{E}}}[\exp((\alpha'-1)D_{\alpha'}(\zeta ,\zeta \ast W))]\leq z}{\inf}D_\alpha(\mu \ast \xi , \nu).$}
    Then, we have: \[D_\alpha(\mu \ast \zeta,\nu \ast \zeta) \leq D_{p(\alpha-1)+1,q(\alpha-1)+1, \zeta}^{(z)}(\mu,\nu) + \frac{\log(z)}{q(\alpha-1)}.\]
    In the case $q = 1$: \[D_\alpha(\mu \ast \zeta,\nu \ast \zeta) \leq D_{\infty,\alpha, \zeta}^{(z)}(\mu,\nu) + \frac{\log(z)}{\alpha-1}.\]
\end{lemma}
This lemma yields a general Wasserstein mechanism that incorporates the noise distribution within the shift. 
\begin{theorem}[Distribution Aware General Wasserstein Mechanism, DAGWM]
\label{EGWM}
Let $f : \mathcal{D} \to \mathbb{R}^d$ be a numerical query and  $\zeta$ noise distribution of $\mathbb{R}^d$. Let $q \geq 1$. For $(s_i,s_j) \in \mathcal{Q}, \theta \in \Theta$, we note $\mu_i^\theta = P(f(X)|s_i,\theta)$. We denote:
\begin{align*}
    \Delta^{\zeta,q,\alpha}_G = \underset{\substack{(s_i,s_j) \in S\\ \theta\in \Theta}}{\max}\text{ }&\underset{P(X,Y) \in \Gamma(\mu_i^\theta,\mu_j^\theta)}{\inf}\\&\mathbb{E}\left[e^{q(\alpha-1)D_{q(\alpha-1)+1}(\zeta ,\zeta \ast (X-Y) )}\right].
\end{align*}
Let $N = (N_1,\dots,N_d) \sim \zeta$ drawn independently of the dataset $X$. Then, $\mathcal{M}(X) = f(X) + N$ satisfies $(\alpha,\frac{\log(\Delta_G^{\zeta,q,\alpha})}{q(\alpha-1)})$-RPP for all $\alpha \in (1,+\infty)$ and $\lim_{\alpha \to + \infty}\frac{\log(\Delta_G^{\zeta,q,\alpha})}{q(\alpha-1)}$-PP. 
\end{theorem}

Leveraging this result allows for the design of mechanisms with sensitivity constrained by $p$-Wasserstein distances ($W_p$). In particular, we will consider noise drawn from 
generalized Cauchy distributions, originally introduced by~\citet{Rider1957}.
\begin{definition}[Generalized Cauchy Distributions]
    Let $k \geq 2, \lambda > 0$. We say that the real random variable $V\sim \GCauchy(\lambda,k)$ if it has the following density:
    \resizebox{\hsize}{!}{$\zeta_{k,\lambda}(x) = \frac{\beta_{k,\lambda}}{((1+(\lambda x)^2)^{k/2}}, x \in \mathbb{R} \text{ and } \int\zeta_{k,\lambda}(x)dx = 1.$}
    The Cauchy distribution is the special case $k = 2$.
\end{definition}
Using generalized Cauchy noise enables to consider $W_p$ shifts while ensuring the existence of moments for large values of $k$.
\begin{corollary}[Cauchy Mechanism]
\label{cauchyMechanism}
    Let $f : \mathcal{D} \to \mathbb{R}^d$ be a numerical query. We denote $Q_\alpha$ the Legendre polynomial of integer index $\alpha > 1$ and $\overline{Q_\alpha}$ as the polynomial derived from $Q_\alpha$ by retaining only its non-negative coefficients. Let $k \geq 2$ and $q \geq 1$ such that $kq(\alpha-1)/2$ is an integer. We note: \resizebox{\hsize}{!}{$\Delta^{dkq(\alpha-1)}_G = \underset{\substack{(s_i,s_j) \in S\\ \theta\in \Theta}}{\max}W_{dkq(\alpha-1)}\left( P(f(X)|s_i,\theta),P(f(X)|s_j,\theta)\right)$}, with  $W_{dkq(\alpha-1)}$ computed with the $l_2$ norm. Then, $\mathcal{M}(X) = f(X) + V$ with $V = (V_1,\dots,V_d) \overset{iid}{\sim} \GCauchy\left(\lambda,k\right)$ is $\Bigg(\alpha,\frac{d\log\frac{\beta_{k,\lambda} \pi}{\lambda}\overline{Q}_{kq(\alpha-1)/2}\left(1 + \left(\frac{\Delta^{dkq(\alpha-1)}_G}{d\lambda}\right)^2\right)}{q(\alpha-1)}\Bigg)$-RPP.
\end{corollary}

In~\hyperref[example]{Example 1}, for $q=d=1$ and $\alpha=k=2$, we have $\Delta_G^{\zeta,2,2} = \sqrt{1+3\rho}$ and noising with $V\sim \Cauchy(\lambda)$ in DAGWM ensures $\left(\alpha,\frac{\log\left(1+\frac{1+3\rho}{\lambda^2}\right)}{\alpha-1}\right)$-RPP, while the GWM for the same noise distribution gives $\left(\alpha,\frac{\log\left(1+\frac{4}{\lambda^2}\right)}{\alpha-1}\right)$-RPP. Hence, in this case DAGWM is better than GWM, as it allows to capture the correlation between the attributes. In the general case, DAGWM consistently outperforms GWM.

\begin{proposition}[Utility of the DAGWM, informal]
\label{DAGWMUtility}
    The DAGWM always offers more utility than the GWM at no additional privacy cost (see Appendix~\ref{appendixDAGWMUtility} for details).
\end{proposition}


\section{Privacy Amplification by Iteration}

\label{compPABI}

Analyzing the privacy guarantees of Pufferfish privacy under composition is known to be challenging \citep{Kifer2014}. While Pufferfish satisfies a form of parallel composition (see Appendix~\ref{PABIAppendix} for the result in RPP), to our knowledge there does not exist any theorem providing mechanism-agnostic guarantees for sequential composition in Pufferfish privacy. As an alternative to composition, we show in this section that RPP is amenable to privacy amplification by iteration, providing a way to analyze iterative gradient descent algorithms for convex optimization.



In differential privacy, privacy amplification by iteration (PABI) allows to evaluate the privacy loss of applying multiple contractive noisy iterations to a dataset and releasing only the output of the last iteration \citep{Feldman2018,Altschuler2022}. PABI has often been employed in private machine learning to analyze the privacy cost of projected noisy stochastic gradient descent (DP-SGD), bypassing the use of composition \citep{Feldman2018}.
However, existing PABI results for differential privacy cannot be used in Pufferfish privacy. These results consider the distribution shift between two processes performed on two neighboring datasets (equal up to one element) and how this additional shift propagates through the rest of the iterations. In Pufferfish, privacy is obtained by conditioning over secrets and the dataset is sampled from an adversary's prior. This means that two datasets with different secrets might share no common elements. Hence, the original worst case PABI analysis must be adapted to account for shifts at each iteration, while measuring these shifts based on the dataset distribution conditioned by the secrets.

We start by defining contractive noisy iterations.

\begin{definition}[Contractive noisy iteration (CNI)]
    Let $\mathcal{Z} \subset \mathbb{R}^d$. Given an initial random state $W_0 \in \mathcal{Z}$, a sequence of random variables $\{X_t\}$, a sequence of contractive maps in their first argument $\psi_t : \mathcal{Z} \times \mathcal{D} \to \mathcal{Z}$ and a sequence of noise distributions $\{\zeta_t\}$, we define the Contractive Noisy Iteration (CNI) by the following update rule: \[W_{t+1} = \psi_{t+1}(W_t,X_{t+1}) + N_{t+1},\] where $N_{t+1} \sim \zeta_{t+1}$. For brevity, we refer to the result $W_T$ of the CNI at the time step $T$ by $\CNI_T(W_0,\{X_t\},\{\psi_t\},\{\zeta_t\})$.
\end{definition}

As opposed to the work of~\citet{Feldman2018}, we make an explicit reference to the dataset distribution modeled by the random sequence $\{X_t\}$ in the CNI definition.
The original PABI analysis leverages a contraction lemma that we need to adapt to the Pufferfish setting. We prove a new contraction lemma which incorporates the $\infty$-Wasserstein distance to take into account the dataset distribution.

\begin{lemma}[Dataset Dependent Contraction lemma]
\label{generalizedContraction}
    Let $\psi$ be a contractive map in its first argument on $(\mathcal{Z},\|\cdot\|)$.  Let $X, X'$ be two r.v's. Suppose that $\sup_{w}W_\infty(\psi(w,X),\psi(w,X'))\leq s$. Then, for $z > 0$:
    \[D^{(z+s)}_{\alpha}(\psi(W,X),\psi(W',X')) \leq D^{(z)}_{\alpha}(W,W').\] 
\end{lemma}

Coupled with the original shift reduction lemma (Lemma~\ref{shiftreduction}), this contraction lemma yields a relaxation of the original PABI bounds, allowing take into account the dataset distribution in the measurement of the shifts.

\begin{theorem}[Dataset Dependent PABI]
\label{newPABI}
    Let $X_T$ and $X'_T$  denote the output of $\CNI_T(W_0,\{\psi_t\},\{\zeta_t\},X)$ and $\CNI_T(W_0,\{\psi_t\},\{\zeta_t\},X')$. Let $s_t =\sup_{w}W_\infty(\psi(w,X_t),\psi(w,X_t'))$. Let $a_1,\dots,a_T$ be a sequence of reals and let $z_t = \sum_{i\leq t}s_i-\sum_{i\leq t}a_i$. If $z_t \geq 0$ for all $t$, then, we have:
    \[D^{(z_T)}_{\alpha}(X_T,X'_T) \leq \sum_{t=1}^{T}R_\alpha(\zeta_t,a_t).\]
\end{theorem}

This new PABI bound allow for an RPP analysis of noisy gradient descent, as developed in the next section.
\section{Applications}
\label{AppsandExps}
\looseness=-1 In this section, we focus on concrete applications of our RPP mechanisms and PABI for specific instantiations. Our generic mechanisms can be hard to compute in the general case, as they rely on the computation of Wasserstein distances between arbitrary distributions. Below, we present specific instances for which the sensitivity of the GWM has a simple closed form. We also apply our generic PABI result to convex optimization, bypassing the lack of adaptive composition theorems and avoiding the cost of group privacy.

\subsection{Weakly Dependent Data}
 The sensitivity $\Delta_G$ of the GWM can be bounded in a straightforward way for \emph{near-independent} data distributions, where the dependence level is quantified via a generalization of Wasserstein dependence metrics \cite{Ozair2019}. Below, we consider $X = (X_1,\dots,X_n) \in \mathcal{X}^n$ with $\mathcal{X} \subset \mathbb{R}^d$ and, for any distribution $\theta \in P(\mathcal{X}^{n})$, we denote by $\theta^\otimes$ the product distribution of the marginals of $\theta$. 
\begin{proposition}
\label{weaklydependent}
    Let $\lambda > 0$, $(\mathcal{S},\mathcal{Q},\Theta)$ a Pufferfish framework. Let $\Delta$ be the sensitivity of a numerical query $f$, denote $\mu_i^\theta = P(f(X)|s_i,\theta)$, and let
    \[\Theta_\lambda = \{\theta \in P(\mathcal{X}^{n}); \textstyle\sup_{s_i \in \mathcal{S}}W_\infty(\mu_i^\theta,\mu_i^{\theta^\otimes}))\leq \lambda\}.\]
    Then, if $\Theta \subseteq \Theta_\lambda$, $\Delta_G\leq 2\lambda + \Delta$. 
\end{proposition}
In other words, for instances with low dependencies, the GWM avoids the extra privacy cost of Group DP.

\subsection{Attribute Privacy}

Attribute privacy~\cite{Zhang2022} is a specific instantiation or Pufferfish.
In this setting, each record $X_i$ in a dataset $X = (X_i^j)_{i\in \Iintv{1,n}, j \in \Iintv{1,m}}$ is viewed as independent, while an adversary possesses prior knowledge, denoted as $\theta$, about the distribution generating each record. The columns, representing each attribute, are denoted by $X^j$. The objective is to reveal the answer of a query $f(X)$ while protecting some summary statistics $g_j(X^j)$ of each attribute $X^j$. A formal definition of attribute privacy is recalled in Appendix~\ref{AppendixAttributeDef}. Below, we show that the GWM can be efficiently computed for Gaussian data, and we also conduct experiments to empirically show that RPP mechanisms have better utility than RDP mechanisms on real datasets.

\paragraph{Special case of Gaussian data.}

 Attribute privacy guarantees for Gaussian priors can be obtained through specific attribute privacy mechanisms~\cite{Zhang2022}. It is also possible to derive closed form bounds with the GWM for linear queries. A simple example is the case of releasing an attribute $X^j$ while protecting another attribute $X^i$.

\begin{proposition}[Attribute privacy with Gaussian data]
    Consider a dataset $X = (X_1^1,\dots,X_n^m)$. Let $j \in \Iintv{1,m}$ be an attribute. We assume that each record $X_i$ is independently sampled from $\theta = \mathcal{N}(\mu,\Sigma)$. Let the secrets $s_{i}^a =\{X_i^j = a\}$, with $a \in K$ and $K$ a compact of $\mathbb{R}^m$, the pairs of secrets $\mathcal{Q} = \{(s_i^a,s_i^b);a,b \in K, i \in  \Iintv{1,n}\}$ and the numerical query $f : x=(x^1,\dots,x^m) \mapsto x^j$. We have
    \[\Delta_G \leq \underset{a,b \in K}{\max}\Var(X_1^i)^{-1}\Cov(X_1^i,X_1^j)\|a-b\|\] for the Pufferfish framework $(\mathcal{S},\mathcal{Q},\{\theta\})$.
\end{proposition}
This proposition shows that if $X^i$ and $X^j$ are weakly correlated, then $\Delta_G$ is small. A more general version of this result can be found in Appendix~\ref{AppendixAttributeGaussian}.

\paragraph{Experiments.}

\cref{tab:attribute_priv} presents some experimental results showing that GWM provides strictly better utility than DP mechanisms in attribute privacy scenarios on three real-world datasets. We obtain lower sensitivities for the DAGWM and the GWM than for DP. Figures, discussions and detailed results can be found in Appendix~\ref{experiments}. 
\begin{table}[ht]
\caption{Sensitivities for DP ($\Delta$), GWM ($\Delta_G$) and Cauchy mechanism ($\Delta_{G,2}$) in attribute privacy scenarios on 3 real datasets.}
\vskip 0.15in
\centering
\label{tab:attribute_priv}
\begin{tabular}{|l|c|c|c|}
\hline
\multicolumn{1}{|c|}{} & \multicolumn{3}{c|}{Sensitivity} \\ \cline{2-4} 
Dataset                & $\Delta$ & $\Delta_G$ & $\Delta_{G,2}$ \\ \hline
Student Scores         & 20      & 8        & 2.76   \\
Heart                  & $\geq $100    & 8        & 7.80    \\
Adult                  & 1      &   1      &     0.42    \\ \hline
\end{tabular}
\end{table}

\subsection{Privacy in Diffusion Processes}

The GWM is also tractable for special cases of temporally correlated data. We consider the setting where one wants to release the output of a query $f(X_{t_1},\dots,X_{t_n})$ performed on a time series $(X_t)_{t\geq 0}$ with $t_1 < \dots < t_n$ while protecting the privacy of the initial point $X_0$. With some assumptions, it is possible to derive contraction results that enable to compute $\Delta_G$. We concentrate on the case where the adversary's prior can be modeled by a Langevin dynamic system.
\begin{proposition}
\label{diffusionprotection}
    Let $V : \mathbb{R}^d\to\mathbb{R}$ such that $\nabla^2 V \succcurlyeq C I_d$. For $\theta_0 \in P(\mathbb{R}^d)$, we note $\theta_t$ the distribution of $X_t$, with $(X_t)_{t\geq0}$ solution of the stochastic differential equation: $dX_t = -\nabla V(X_t)dt+\sqrt{2}dB_t$, where $(B_t)_{t\geq0}$ is a brownian motion. We note $\theta_{t_1,\dots,t_n}$ the distribution generating $X = (X_{t_1},\dots,X_{t_n})$ from the distribution of $(X_t)_{t\geq0}$. We consider the secrets $s^a =\{X_0= a\}$, with $a \in K$ and $K$ a compact of $\mathbb{R}^d$, and the pairs of secrets $\mathcal{Q} = \{(s^a,s^b);a,b \in K\}$. Then, the GWM of any $L$-Lipschitz query $f$ performed on $X$ has a sensitivity for the $l_1$ norm:
    \[\Delta_G \leq L \text{\emph{Diam}}(K) \sum_{i=1}^n \exp(-2Ct_i) \] for the Pufferfish framework $(\mathcal{S},\mathcal{Q},\{\theta_{t_1,\dots,t_n}\})$.  
\end{proposition}
This result demonstrates that for some well-behaved diffusion processes,
the GWM drastically mitigates the sensitivity compared to the Group DP scenario. This reduction is indicated by an exponential term which depends on the convexity of $V$ and the released timestamps.

\subsection{Application of PABI to Convex Optimization}
\label{sectionPABIConvex}
We show an application of our general PABI result (Theorem~\ref{newPABI}) to the RPP analysis of the celebrated DP-SGD algorithm for private machine learning.

 Let $m,d, T > 0$. Let $(\mathcal{S},\mathcal{Q},\Theta)$ be a Pufferfish framework. We note $\mathcal{X}$ the set of values taken by the elements of the dataset. Let the secrets $s_i^a = \{X_i=a\},s_i^b = \{X_i=b\}$, $i\in \Iintv{1,T}, a,b \in \mathcal{X}$. We note $X = (X_1,\dots,X_T) \sim P(X|s_i^a,\theta)$ and $X' = (X'_1,\dots,X'_T) \sim P(X|s_i^b,\theta)$. We assume that $\mathcal{X} \subset \mathbb{R}^m$.
    Let $f: \mathbb{R}^d \times \mathcal{X} \to \mathbb{R}$ be an objective function. We assume that $f$ is convex, $L$-Lipschitz in its first argument, $\beta$-smooth in its second argument (see Appendix~\ref{AppendixPABISetup} for definitions) and $f$ satisfies the following condition: $\forall x_1,x_2\in \mathcal{X},w_1\in \mathbb{R}^d, \exists C_{w_1} > 0$ such as : \[\|\nabla_wf(w_1,x_1)-\nabla_wf(w_1,x_2)\| \leq C_{w_1}\|x_1-x_2\|.\]

The last assumption, which is used in the adversarial training literature~\citep[see e.g.,][]{Liu2020}, is satisfied in certain simple settings as linear regression. It enables to take into account the distribution of the gradients as a function of the distribution of the data in our PABI analysis. 
Let $\Pi : \mathbb{R}^d \to \mathbb{R}^d$ be a projection over a compact $\mathcal{K}\subset \mathbb{R}^d$ and $\eta > 0$ such that $\eta < 2/\beta$. By Proposition 18 of~\citet{Feldman2018}, the weight update function: 
$\psi : \mathbb{R}^d\times\mathcal{X} 
\to \mathbb{R}^d,
    (v,x) \mapsto \Pi(v-\eta\nabla_wf(v,x))$
is contractive. Let $W_0=W_0' \in \mathcal{K}$ be the initial weight and $\zeta = \mathcal{N}(0,\sigma^2I_d)$ be a noise distribution. We note $(N_1,\dots,N_T) \sim \zeta^{\otimes T} $ and for all $t \in \Iintv{1,T}$, $W_{t} = \psi(W_{t-1}, X_t) + N_t$, $W'_{t} = \psi(W'_{t-1}, X'_t) + N_t$, as in DP-SGD.
Then, we note $s_t =\eta\sup_{v\in  \mathcal{K}}W_\infty(\nabla_wf(v,X_t),\nabla_wf(v,X'_t))$.
As an example of application of Theorem~\ref{newPABI}, taking $(a_t) = (s_t)$, we have:
\[D_{\alpha}(W_T,W'_T) \leq \frac{\alpha\eta^2}{2\sigma^2 }\sum_{t=1}^{T}\min(2L,\sup_{v\in \mathcal{K}}C_vW_\infty(X_t,X'_t))^2.\]
To interpret this formula, we can look at some extreme cases.
If the adversary has a prior of high correlations, such as for example $X_1 = \dots = X_t$, $X'_1 = \dots = X'_t$, we get:
\[D_{\alpha}(W_T,W'_T) \leq \frac{T\alpha\eta^2}{2\sigma^2 }\min(2L,\|a-b\|\sup_{v\in \mathcal{K}}C_v)^2,\]
which is no better than the group privacy analysis. On the other hand, when data points are independent as in differential privacy, we get:
\[D_{\alpha}(W_T,W'_T) \leq \frac{\alpha\eta^2}{2\sigma^2 }\min(2L,\|a-b\|\sup_{v\in \mathcal{K}}C_v)^2 .\]
In this case, the upper bound is independent of $T$ and we thus obtain much better results than with group privacy. In fact, our result is general enough to recover the original results of~\citet{Feldman2018} for DP-SGD as a special case.
\begin{remark}[DP as a special case, informal]
    Theorem~\ref{newPABI} allows to recover the same privacy bounds as Theorem 23 of~\citet{Feldman2018} (see Appendix~\ref{DPapplication} for details).
\end{remark}
In the Gaussian case, our results allow to derive PABI bounds that explicitly depend on correlations in the dataset.
\begin{proposition}
\label{PABIGaussian}
    Assume that the adversary has a Gaussian prior $\theta$. Then,
    \begin{align*}
        &D_{\alpha}(W_T,W'_T) \leq \frac{\alpha\eta^2}{2\sigma^2 }\big(\min(2L,\sup_{v\in \mathcal{K}}C_v\|a-b\|)^2+\\&\sum_{t\neq i}^{T}\min(2L,\sup_{v\in \mathcal{K}}C_v \|\Cov(X_t,X_i)\Cov(X_i)^{-1}(a-b)\|)^2\big).
    \end{align*}
\end{proposition}

\looseness=-1 This result is the sum of two terms: the first one is the same as for DP (i.e., the case where data points are independent), while the second one accounts for the dependence by summing, for each step $t$, the correlation between $X_t$ and $X_i$. 

This bound can be improved when $W_\infty(X_t,X_t')$ is non-increasing,  leading to settings where the privacy loss converges to $0$ as $T \to + \infty$. This consideration is discussed and illustrated numerically in Appendix~\ref{appendixPABIdecreasing}.

\section{Conclusion}
We presented a new framework, called Rényi Pufferfish privacy, which extends the original Pufferfish privacy definition. We designed general additive noise mechanisms for achieving (approximate) Rényi Pufferfish privacy and discussed their applicability for specific instantiations. As a way to use Pufferfish privacy to analyze sequential algorithms, we derived a privacy amplification by iteration result which allows to bypass the lack of adaptive composition theorems. We put forward a first application of this analysis for convex optimization with gradient descent, allowing the integration of Pufferfish in machine learning algorithms. Potential areas for future work include a tighter PABI analysis with other shift reduction lemmas, and a numerical analysis of Rényi Pufferfish privacy mechanisms to optimize utility in more complex practical use-cases.

\section*{Impact Statement}

This paper presents work whose goal is to advance privacy in machine learning, offering tools to make it more secure. We provide methods to protect specific types of secrets and to manage correlations in datasets, which are frequently found in practice. Properly designed Pufferfish instantiations can provide greater utility than usual Group Differential Privacy mechanisms. However, the data curator must carefully design its Pufferfish instantiation in order to ensure adequate robust privacy protection.

\bibliography{example_paper}
\bibliographystyle{icml2024}

\newpage
\appendix
\onecolumn
This appendix provides some useful background, as well as more detailed versions of our results, along with their proofs.

\section{Properties of Rényi Pufferfish Privacy (Section~\ref{sectionrenyi})}
\subsection{Definitions}
\label{AppendixRDPRefs}
 Rényi differential privacy relies on Rényi divergences, which are defined as follows.
 \begin{definition}
     Let $\mu$ and $\nu$ be two distributions on a measurable space $(E,\mathcal{A})$ and $\alpha > 1$. We define the Rényi divergence of order $\alpha$ between $\mu$ and $\nu$ as:
     \[D_\alpha(\mu,\nu) = \frac{1}{\alpha-1}\log\mathbb{E}_{x \sim \nu}\left[\left(\frac{\mu(x)}{\nu(x)}\right)^\alpha\right].\]
     The definition extends to the case $\alpha = +\infty$ by continuity.
 \end{definition}
 \begin{definition}[Rényi differential privacy, RDP \citep{Mironov2017}]
 \looseness=-1 Let $\alpha > 1$ and $\varepsilon \geq 0$. A randomized algorithm $\mathcal{M} \colon \mathcal{D} \to \mathcal{R}$ satisfies $(\alpha,\varepsilon)$-Rényi differential privacy if for any two adjacent datasets $X_1, X_2 \in \mathcal{D}$ differing by one element, it holds:\[D_\alpha\left(P(\mathcal{M}(X_1)),P(\mathcal{M}(X_2)\right) \leq \varepsilon.\]
 \end{definition}
\subsection{Proof of Proposition~\ref{postproc}}

\begingroup
\def\theproposition{\ref{postproc}}
\begin{proposition}[Post-processing]
    Let $\mathcal{M}_1$ be a randomized algorithm and $\mathcal{M}$ be $(\alpha,\varepsilon)$-RPP. Then,
    \begin{align*}
        &D_\alpha\left(P\left(\mathcal{M}_1(\mathcal{M}(X)) \mid s_{i}, \theta\right), P\left(\mathcal{M}_1(\mathcal{M}(X)) \mid s_{j}, \theta\right)\right) \\ &\leq D_\alpha\left(P\left(\mathcal{M}(X) \mid s_{i}, \theta\right), P\left(\mathcal{M}(X) \mid s_{j}, \theta\right)\right) \leq \varepsilon.
    \end{align*}
\end{proposition}
\addtocounter{proposition}{-1}
\endgroup

\begin{proof}
    Let $\mathcal{S},\mathcal{Q},\Theta)$ be a Pufferfish framework. Let $(s_i,s_j)\in \mathcal{Q}, \theta \in \Theta$, $\alpha > 1$ and $ \epsilon > 0$. Let $\mathcal{M}_1$ be a randomized algorithm and $\mathcal{M}$ satisfying $(\alpha,\varepsilon)$-RPP. Then,
    \begin{align*}
        &D_\alpha\left(P\left(\mathcal{M}(X) \mid s_{i}, \theta\right), P\left(\mathcal{M}(X) \mid s_{j}, \theta\right)\right) = \mathbb{E}_{Z \sim P\left(\mathcal{M}(X) \mid s_{j},\theta\right)}\left[\left(\frac{P\left(\mathcal{M}(X)=Z \mid s_{i}, \theta\right)}{P\left(\mathcal{M}(X)=Z \mid s_{j}, \theta\right)}\right)^\alpha\right] \\      &= \mathbb{E}_{(Z',Z) \sim (P\left(\mathcal{M}_1(\mathcal{M}(X)) \mid s_{j},\theta\right), P\left(\mathcal{M}(X) \mid s_{j},\theta\right))} \left[\left(\frac{P\left(\mathcal{M}(X)=Z \mid s_{i}, \theta\right)}{P\left(\mathcal{M}(X)=Z \mid s_{j}, \theta\right)}\right)^\alpha\right] \\ 
        &= \mathbb{E}_{(Z',Z) \sim (P\left(\mathcal{M}_1(\mathcal{M}(X)) \mid s_{j},\theta\right), P\left(\mathcal{M}(X) \mid s_{j},\theta\right))}\left[\left(\frac{P\left(\mathcal{M}(X)=Z \mid s_{i}, \theta\right)}{P\left(\mathcal{M}(X)=Z \mid s_{j}, \theta\right)}\frac{P\left(\mathcal{M}_1(\mathcal{M}(X))=Z' \mid\mathcal{M}(X)=Z\right)}{P\left(\mathcal{M}_1(\mathcal{M}(X))=Z'\mid \mathcal{M}(X)=Z\right)}\right)^\alpha\right] \\ 
        &= \mathbb{E}_{Z' \sim P\left(\mathcal{M}_1(\mathcal{M}(X)) \mid s_{j},\theta\right)}\left[\mathbb{E}_{Z \sim P\left(\mathcal{M}(X) \mid\mathcal{M}_1(\mathcal{M}(X)), s_{j},\theta\right)} \left[\left(\frac{P\left(\mathcal{M}_1(\mathcal{M}(X))=Z' ,\mathcal{M}(X)=Z\mid  s_{i}, \theta\right)}{P\left(\mathcal{M}_1(\mathcal{M}(X))=Z', \mathcal{M}(X)=Z\mid s_{j}, \theta\right)}\right)^\alpha\right]\right] \\
        &\geq \mathbb{E}_{Z' \sim P\left(\mathcal{M}_1(\mathcal{M}(X)) \mid s_{j},\theta\right)}\left[\left(\mathbb{E}_{Z \sim P\left(\mathcal{M}(X) \mid\mathcal{M}_1(\mathcal{M}(X)), s_{j},\theta\right)} \left[\frac{P\left(\mathcal{M}_1(\mathcal{M}(X))=Z' ,\mathcal{M}(X)=Z\mid  s_{i}, \theta\right)}{P\left(\mathcal{M}_1(\mathcal{M}(X))=Z', \mathcal{M}(X)=Z\mid s_{j}, \theta\right)}\right]\right)^\alpha\right] \substack{\text{Jensen} \\ \text{inequality}} \\
        &= \mathbb{E}_{Z' \sim P\left(\mathcal{M}_1(\mathcal{M}(X)) \mid s_{j},\theta\right)}\left[\left(\frac{P\left(\mathcal{M}_1(\mathcal{M}(X))=Z \mid  s_{i}, \theta\right)}{P\left(\mathcal{M}_1(\mathcal{M}(X))=Z\mid s_{j}, \theta\right)}\right)^\alpha\right] \\
        &= D_\alpha\left(P\left(\mathcal{M}_1(\mathcal{M}(X))\mid  s_{i}, \theta\right), P\left(\mathcal{M}_1(\mathcal{M}(X)) \mid s_{j}, \theta\right)\right). 
    \end{align*}
    Thus, \[D_\alpha \left(P\left(\mathcal{M}_1(\mathcal{M}(X))\mid  s_{i}, \theta\right), P\left(\mathcal{M}_1(\mathcal{M}(X)) \mid s_{j}, \theta\right)\right) \leq D_\alpha\left(P\left(\mathcal{M}(X) \mid s_{i}, \theta\right), P\left(\mathcal{M}(X) \mid s_{j}, \theta\right)\right) \leq \varepsilon.\qedhere\]
\end{proof}

\subsection{Proof of Proposition~\ref{RPPtoPP}}
\begingroup
\def\theproposition{\ref{RPPtoPP}}
\begin{proposition}[RPP implies PP]

If $\mathcal{M}$ is $(\alpha,\varepsilon)$-RPP, it also satisfies $\big(\varepsilon + \frac{\log(1/\delta)}{\alpha-1},\delta\big)$-PP for all $\delta \in (0,1)$.
\end{proposition}
\addtocounter{proposition}{-1}
\endgroup
\begin{proof}
    The proof technique of~\cite{Mironov2017} remains applicable in the context of Rényi Pufferfish privacy. For clarity and completeness, we showcase it here. Let $\varepsilon \geq 0, \alpha >1$. Let $(\mathcal{S},\mathcal{Q},\Theta)$ be a Pufferfish privacy framework and $\mathcal{M}$ an $(\alpha,\varepsilon)$-RPP mechanism. Let $\delta \in (0,1)$, $\theta \in \Theta$, $(s_i,s_j) \in \mathcal{Q}$ and $z \in Range(\mathcal{M)}$. Then, we have:
    \[P\left(\mathcal{M}(X)=z|s_i,\theta\right)^\alpha \leq e^{(\alpha-1)D_\alpha\left(P\left(\mathcal{M}(X)|s_i,\theta\right),P\left(\mathcal{M}(X)|s_j,\theta\right)\right)}P\left(\mathcal{M}(X)=z|s_j,\theta\right)^{\alpha-1} \leq e^{\varepsilon (\alpha-1)} P\left(\mathcal{M}(X)=z|s_j,\theta\right)^{\alpha-1},\]
    where the first inequality is obtained by Hölder inequality applied to the functions $\left(\frac{f^\alpha}{g^{\alpha-1}}\right)^{\frac{1}{\alpha}}$ and $g^\frac{\alpha-1}{\alpha}$. We then consider two cases:
    \begin{itemize}
        \item Case 1: $e^\varepsilon P\left(\mathcal{M}(X)=z|s_j,\theta\right) \leq \delta^{\frac{\alpha}{\alpha-1}}$. Then, $P\left(\mathcal{M}(X)=z|s_i,\theta\right) \leq \delta \leq e^{\varepsilon+\frac{\log(1/\delta)}{\alpha-1}}P\left(\mathcal{M}(X)=z|s_j,\theta\right) + \delta$.
        \item Case 2: $e^\varepsilon P\left(\mathcal{M}(X)=z|s_j,\theta\right) > \delta^{\frac{\alpha}{\alpha-1}}$. Then, 
        \begin{align*}
            P\left(\mathcal{M} (X)=z|s_i,\theta\right) &\leq \left(e^{\varepsilon} P\left(\mathcal{M}(X)=z|s_j,\theta\right)\right)\left(e^{\varepsilon} P\left(\mathcal{M}(X)=z|s_j,\theta\right)\right)^{\frac{-1}{\alpha}}\\ &\leq e^{\varepsilon} P\left(\mathcal{M}(X)=z|s_j,\theta\right)\delta^{\frac{-1}{\alpha-1}}\\ &\leq e^{\varepsilon+\frac{\log(1/\delta)}{\alpha-1}} P\left(\mathcal{M}(X)=z|s_j,\theta\right) + \delta.\qedhere
        \end{align*}
    \end{itemize}
    
\end{proof}

\subsection{Guarantees Against Close Adversaries}
\label{closeappendix}

\subsubsection{Original result from Song et al.~\yrcite{Song2017}}

For completeness, we recall here the original theorem from Song et al.~\yrcite{Song2017} on the robustness of the Pufferfish privacy framework.

\begin{theorem}[Protection against close adversaries \cite{Song2017}]
\label{closeadvpp}
   Let $\mathcal{M}$ be a mechanism that satisfies $\varepsilon$-PP in a framework $(\mathcal{S},\mathcal{Q},\Theta)$. Let $\theta' \notin \Theta$ and
   \begin{align*}
       \Delta = \underset{\theta \in \Theta}{\inf}\underset{s_i \in \mathcal{Q}}{\sup}\max \{&D_{\infty}\left(P(X|s_i,\theta),P(X|s_i,\theta')\right),\\&D_{\infty}\left(P(X|s_i,\theta),P(X|s_i,\theta')\right)\}. 
   \end{align*}
   Then, $\mathcal{M}$ is $(\varepsilon+2\Delta)$-PP for the framework $(\mathcal{S},\mathcal{Q},\Theta')$ with $\Theta' = \Theta \cup \{\theta'\}$.
\end{theorem}

\subsubsection{Protection against close adversaries in the RPP framework}

\begin{theorem}[RPP protection against close adversaries]
\label{closeadvthm}
    Let $p,q,r > 0$ such that $\frac{1}{p}+\frac{1}{q}+\frac{1}{r}=1$, and let $\mathcal{M}$ be a mechanism that satisfies $(q(\alpha-1/p),\varepsilon)$-RPP in a framework $(\mathcal{S},\mathcal{Q},\Theta)$. Let $\theta' \notin \Theta$ and 
    \begin{align*}
        \Delta^1_p &= \underset{\theta \in \Theta}{\inf}\underset{s_i \in S}{\sup} D_{\alpha p}\left(P(X|s_i,\theta'),P(X|s_i,\theta)\right),\\
        \Delta^2_r &= \underset{\theta \in \Theta}{\inf}\underset{s_i \in S}{\sup} D_{(\alpha-1)r+1}\left(P(X|s_i,\theta),P(X|s_i,\theta')\right).
    \end{align*}
    Then, for all $\alpha \in (1,\infty)$, $\mathcal{M}$ satisfies:
    \[\bigg(\alpha,\Big(1 + \frac{1}{r(\alpha - 1)}\Big)\varepsilon + \bigg(1+ \frac{\frac{1}{r} + \frac{1}{q}}{\alpha-1}\bigg)\Delta^1_p + \Delta^2_r\bigg)\text{-RPP}\] for $(\mathcal{S},\mathcal{Q},\Theta')$ with $\Theta' = \Theta \cup \{\theta'\}$.
\end{theorem}

\begin{proof}
    Let $(\mathcal{S},\mathcal{Q},\Theta)$ be a Pufferfish privacy instance and $\mathcal{M}$ a randomized mechanism. Let $\theta' \notin \Theta$ and $p,q,r > 0$ such that $\frac{1}{p}+\frac{1}{q}+\frac{1}{r}=1$. Let $s_i, s_j \in \mathcal{Q}$. We have: 
    \begin{alignat*}{2}
        &\exp&&{(\alpha-1)D_{\alpha}\left(P(\mathcal{M}(X)|s_i,\theta'),P(\mathcal{M}(X)|s_j,\theta')\right)}\\
        &= &&\int \frac{P(\mathcal{M}(X)=z|s_i,\theta')^\alpha}{P(\mathcal{M}(X)=z|s_j,\theta')^{\alpha-1}}dz\\
        &= &&\int \frac{P(\mathcal{M}(X)=z|s_i,\theta')^\alpha}{P(\mathcal{M}(X)=z|s_i,\theta)^{\alpha-1/p}}\frac{P(\mathcal{M}(X)=z|s_i,\theta)^{\alpha-1/p}}{P(\mathcal{M}(X)=z|s_i,\theta)^{\alpha-1/p-1/q}}\frac{P(\mathcal{M}(X)=z|s_j,\theta)^{\alpha-1/p-1/q}}{P(\mathcal{M}(X)=z|s_j,\theta')^{\alpha-1}}dz\\
        &\leq &&\left(\int \frac{P(\mathcal{M}(X)=z|s_i,\theta')^{\alpha p}}{P(\mathcal{M}(X)=z|s_i,\theta)^{\alpha p-1}} dz \right)^\frac{1}{p}  \cdot\left(\int \frac{P(\mathcal{M}(X)=z|s_i,\theta)^{q(\alpha-1/p)}}{P(\mathcal{M}(X)=z|s_i,\theta)^{q(\alpha-1/p)-1}} dz\right)^\frac{1}{q} \\ & \cdot &&\left(\int \frac{P(\mathcal{M}(X)=z|s_j,\theta)^{\alpha-1/p-1/q}}{P(\mathcal{M}(X)=z|s_j,\theta')^{\alpha-1}}dz\right)^\frac{1}{r}\\
        & \leq && \exp {(\alpha-1/p)D_{\alpha p}\left(P(\mathcal{M}(X)|s_i,\theta'),P(\mathcal{M}(X)|s_i,\theta)\right)} \\& && + \exp  \left({(\alpha-1 + 1/r)D_{q(\alpha-1/p)}\left(P(\mathcal{M}(X)|s_i,\theta),P(\mathcal{M}(X)|s_j,\theta)\right)} \right)\\& && +  \exp  \left({(\alpha-1)D_{(\alpha-1)r+1}\left(P(\mathcal{M}(X)|s_j,\theta),P(\mathcal{M}(X)|s_j,\theta')\right)} \right)
    \end{alignat*}
    by using the generalized Hölder inequality: for $p,q,r,t > 0$ such that $\frac{1}{p}+ \frac{1}{q} + \frac{1}{r} = \frac{1}{t}$ and $f\in L^p,g\in L^q, h \in L^r$, \[\|fgh\|_t \leq \|f\|_p\|g\|_q\|h\|_r.\]
    Then, the post-processing property of RPP (Proposition~\ref{postproc}) gives the result.
\end{proof}
This theorem employs $\alpha$-Rényi divergences and can be viewed as a generalization of the result of Song et al.~\yrcite{Song2017}, which we recover as a special case for $\alpha = +\infty$.
Note that neither Theorem~\ref{closeadvthm} nor the original result of Song et al.~\yrcite{Song2017} exploit the characteristics of the particular mechanism $\mathcal{M}$ of interest in the quantification of the additional privacy loss. As a matter of fact, it is likely that a mechanism with large variance would yield more robust guarantees. Interestingly, we can address this issue by refining our result to additive noise mechanisms using the shift reduction lemma.
\subsubsection{Refinement of Theorem~\ref{closeadvthm} for Additive Mechanisms}

Leveraging the shift reduction lemma (Lemma~\ref{shiftreduction}), we refine Theorem~\ref{closeadvthm} for additive mechanisms.
\begin{theorem}[RPP protection against close adversaries for additive noise mechanisms]
    \label{closeadvadd}
    Let $p,q,r > 0$ such that $\frac{1}{p}+\frac{1}{q}+\frac{1}{r}=1$. Let $f : \mathcal{D} \to \mathbb{R}^d$ be a numerical query. Let $\mathcal{M}(X) = f(X) + N$ with $N \sim \zeta$ be an additive noise mechanism that satisfies $(q(\alpha-1/p),\varepsilon)$-RPP for $(\mathcal{S},\mathcal{Q},\Theta)$. Let $\theta' \notin \Theta$ and
    \[\Delta_{\theta'} = \underset{\theta \in \Theta}{\inf}\underset{s_i \in S}{\sup} W_\infty\left(P(f(X)|s_i,\theta'),P(f(X)|s_i,\theta)\right).\]
    Then, for all $\alpha \in (1,\infty)$ and denoting \[K = \left(1+ \frac{\frac{1}{r} + \frac{1}{q}}{\alpha-1}\right)R_{\alpha p}(\zeta,\Delta_{\theta'}) + R_{(\alpha-1)r+1}(\zeta,\Delta_{\theta'}),\]
    $\mathcal{M}$ satisfies:
    \[\left(\alpha,\left(1 + \frac{1}{r(\alpha-1)}\right)\varepsilon + K\right)\text{-RPP}\] for $(\mathcal{S},\mathcal{Q},\Theta')$ with $\Theta' = \Theta \cup \{\theta'\}$.
\end{theorem}
This theorem enables us to take into account the characteristics of the  mechanism when examining the robustness of a RPP instance. We illustrate this below with the Gaussian mechanism.

\begin{corollary}[RPP protection against close adversaries for the Gaussian mechanism]
    We note $I_d$ the identity matrix of size $d$. Let $p,q,r > 0$ such that $\frac{1}{p}+\frac{1}{q}+\frac{1}{r}=1$. Let $f : \mathcal{D} \to \mathbb{R}^d$ be a numerical query. Let $\mathcal{M}(X) = f(X) + N$ with $N \sim \mathcal{N}\big(0,\frac{q(\alpha-1/p)\Delta_G^2}{2\varepsilon}I_d\big)$, where $\Delta_G$ is defined in Theorem~\ref{GWM}. Let $\theta' \notin \Theta$. 
    
    Then, for all $\alpha \in (1,\infty)$, $\mathcal{M}$ satisfies:
    \[\left(\alpha,\left(\left(1 + \frac{1}{r(\alpha-1)}\right) + \left(\alpha \left(p+ \frac{p-1}{\alpha-1}\right)  + (\alpha-1)r+1\right)\frac{\Delta_{\theta'}^2}{\Delta_G^2}\frac{1}{q(\alpha-1/p)}\right)\varepsilon\right)\text{-RPP}\] for $(\mathcal{S},\mathcal{Q},\Theta')$ with $\Theta' = \Theta \cup \{\theta'\}$.
\end{corollary}
One can see that the additive penalty vanishes proportionally to $\frac{1}{\Delta_G^2}$. It establishes a trade-off between the utility of the mechanism and the robustness of the Pufferfish privacy framework when designing $\Theta$. Remarkably, this consideration could not have been derived from our Theorem~\ref{closeadvthm} for RPP nor from the original result from Song et al.~\yrcite{Song2017} (Theorem~\ref{closeadvpp}).

\section{General Wasserstein Mechanism (Section~\ref{sectionwassersteinmechanism})}

\subsection{Proof of Theorem~\ref{GWM}}
\begingroup
\def\thetheorem{\ref{GWM}}
\begin{theorem}[General Wasserstein mechanism, GWM]
Let $f : \mathcal{D} \to \mathbb{R}^d$ be a numerical query and denote:

\[\Delta_G = \underset{\substack{(s_i,s_j) \in S\\ \theta\in \Theta}}{\max} W_\infty \left( P(f(X)|s_i,\theta),P(f(X)|s_j,\theta)\right).\]

Let $N = (N_1,\dots,N_d) \sim \zeta$, where $N_1,\dots,N_d$ are iid real random variables independent of the data $X$. Then, $\mathcal{M}(X) = f(X) + N$ satisfies $(\alpha,R_\alpha(\zeta,\Delta_G))$-RPP for all $\alpha \in (1,+\infty)$ and $R_\infty(\zeta,\Delta_G)$-PP. 
\end{theorem}
\addtocounter{theorem}{-1}
\endgroup
\begin{proof}
Let $(\mathcal{S},\mathcal{Q},\Theta)$ be a Pufferfish privacy instance. Let $f : \mathcal{D} \to \mathbb{R}^d$ be a numerical query and denote:

\[\Delta_G = \underset{\substack{(s_i,s_j) \in S\\ \theta_i\in \Theta}}{\max} W_\infty \left( P(f(X)|s_i,\theta),P(f(X)|s_j,\theta)\right).\]

Let $N = (N_1,\dots,N_d) \sim \zeta$, where $N_1,\dots,N_d$ are iid real random variables independent of the data $X$. We use the abuse of notation $D_\alpha(X_{|E},Y_{|E}) = D_\alpha(P(X|E),P(Y|E))$. Let $\alpha > 1$, $z > 0$, $(s_i,s_j)\in \mathcal{Q}$ and $\theta \in \Theta$. By the shift reduction lemma (Lemma~\ref{shiftreduction}), we have:
\[D_\alpha\left(\left(f(X) + N \right)_{\mid s_i, \theta},\left(f(X) + N \right)_{\mid s_j, \theta}\right) \leq D_\alpha^{(z)}\left(f(X) _{\mid s_i, \theta},f(X)_{\mid s_j, \theta}\right) + R_\alpha(\zeta,z).\]
By definition, \[D_\alpha^{(z)}\left(f(X) _{\mid s_i, \theta},f(X)_{\mid s_j, \theta}\right) = \underset{W \in \mathcal{P}(\mathbb{R}^d); W_\infty\left(W,f(X)_{\mid s_i, \theta}\right) \leq z}{\inf}D_\alpha\left(W,f(X)_{\mid s_j, \theta}\right),\] and \[D_\alpha^{(W_\infty \left( P(f(X)|s_i,\theta),P(f(X)|s_j,\theta)\right))}\left(f(X) _{\mid s_i, \theta},f(X)_{\mid s_j, \theta}\right) = 0.\] Then, \[D_\alpha\left(\left(f(X) + N \right)_{\mid s_i, \theta},\left(f(X) + N \right)_{\mid s_j, \theta}\right) \leq  R_\alpha\left(\zeta,W_\infty \left( P(f(X)|s_i,\theta),P(f(X)|s_j,\theta)\right)\right) \leq R_\alpha\left(\zeta,\Delta_G\right).\qedhere\]
    
\end{proof}

\subsection{Proof of Corollary \ref{wassersteinusualdistrib}}
In order to compute $R_\alpha$ for usual distributions, we use the following result. 
\begin{lemma}[$R_\alpha$ calculation criterion]
\label{sym}
    Let $\alpha > 1$, $d \in \mathbb{N}^*$. Let $\zeta$ be a distribution on $\mathbb{R}$. If $\zeta = e^{-g}$ is even, non-decreasing on $\mathbb{R}^+$ and $z \mapsto D_\alpha(\zeta_{-g^{-1}(z)},\zeta)$ is convex on $\mathbb{R}^+$, then $\sup_{g^{-1}(\sum_{i=1}^d g(x_i))\leq z}D_\alpha(\zeta^{\otimes d}_{-x},\zeta^{\otimes d}) = D_\alpha(\zeta_{-z},\zeta)$, where $\zeta^{\otimes d}$ is the joint distribution of $d$ independent random variables drawn from $\zeta$, and $g^{-1}$ is the inverse of $g$ on $\mathbb{R}^+$.
\end{lemma}

\begin{proof}
    We start by proving the following characterization of convex functions: if $g$ is a convex function of $\mathbb{R}$, then for $z_1 \leq z_2$, the function $x \mapsto g(x-z_1)-g(x-z_2)$ is non-decreasing. A similar statement and its proof can be found in~\cite{Saumard2014}. 
    
    Let $x \leq x'$. By taking $\lambda = \frac{x-x'}{x-x'+z_1-z_2}$, we have: 

        \[x-z_1 = (1-\lambda) (x'-z_1) + \lambda(x-z_2) \text{ and }
        x'-z_2 = \lambda (x'-z_1) + (1-\lambda)(x-z_2).\]

    By convexity:
        \[g(x-z_1) \leq (1-\lambda) g(x'-z_1) + \lambda g(x-z_2) \text{ and }
        g(x'-z_2) = \lambda g(x'-z_1) + (1-\lambda)g(x-z_2).\]
    Then, $g(x'-z_1)-g(x'-z_2) \geq g(x-z_1)-g(x-z_2) $.

    We also prove that if $g$ is convex, $g(0) = 0$ and $y_1,\dots,y_d \in \mathbb{R}^+$, $g(\sum_{i=1}^d y_i) \geq \sum_{i=1}^d g(y_i)$.
    We prove this assertion by induction. Obviously, $g(y_1) \geq g(y_1)$. Also, $\sum_{i=1}^d y_i \geq 0$. Then, $x\mapsto g(x+\sum_{i=1}^d y_i)-g(x)$ is non-decreasing. For $y_{d+1} \geq 0$, we have $g(\sum_{i=1}^{d+1} y_i)-g(y_{d+1})\geq g(\sum_{i=1}^d y_i)-g(0)$, and by induction $g(\sum_{i=1}^{d+1} y_i) \geq \sum_{i=1}^{d+1} g(y_i)$.

    Then, if $z \mapsto D_\alpha(\zeta_{-g^{-1}(z)},\zeta)$ is convex, for $x \in \mathbb{R}^{d+}$, we have:
    \[D_\alpha(\zeta^{\otimes d}_{-x},\zeta^{\otimes d}) = \sum_{i=1}^d D_\alpha(\zeta_{-g^{-1}(g(x_i))},\zeta) \geq   D_\alpha\left(\zeta_{-g^{-1}(\sum_{i=1}^d g(x_i))},\zeta\right),\]
     which proves the statement.
\end{proof}

As we only considered shifts in $\mathbb{R}^{d+}$, we also need $z \mapsto D_\alpha(\zeta_{-z},\zeta)$ to be even, which is achieved for symmetrical densities: if $\zeta$ is symmetric, $z \mapsto D_\alpha(\zeta_{-z},\zeta)$ is also symmetric.
In fact: \[\int_{-\infty}^{+\infty}\frac{\zeta(x-z)^\alpha}{\zeta(x)^{\alpha-1}}dx = \int_{-\infty}^{+\infty}\frac{\zeta(x+z)^\alpha}{\zeta(x)^{\alpha-1}}dx\] by symmetry and changes of variable.

\begingroup
\def\thecorollary{\ref{wassersteinusualdistrib}}
\begin{corollary}[Privacy guarantees for usual noise distributions]
We note $I_d$ the identity matrix of size $d$. Plugging the expressions of $R_\infty(\zeta,z)$ and $R_\alpha(\zeta,z)$ for Laplacian and Gaussian distributions, we obtain:
\begin{itemize}
    \item $\mathcal{M}(X) = f(X) + N$ with $N \sim \mathcal{N}\big(0,\frac{\alpha \Delta_G^2}{2\varepsilon}I_d\big)$ and $\Delta_G$ computed on the $l_2$ norm is $(\alpha, \varepsilon)$-RPP.
    \item $\mathcal{M}(X) = f(X) + L$ with $L \sim \Lap(0,\rho I_d)$ and $\Delta_G$ computed on the $l_1$ norm is $\Big(\alpha,\frac{1}{\alpha-1}\log\big(\frac{\alpha}{2\alpha-1}e^{\Delta_G(\alpha-1)/\rho} + \frac{\alpha-1}{2\alpha-1}e^{-\Delta_G \alpha/ \rho} \big)\Big)$-RPP.
    \item $\mathcal{M}(X) = f(X) + L$ with $L \sim \Lap\left(0,\frac{\Delta_G}{\varepsilon}I_d\right)$ with $\Delta_G$ computed on the $l_1$ norm is $\varepsilon$-PP.
\end{itemize}
\end{corollary}
\addtocounter{corollary}{-1}
\endgroup
\begin{proof}
The result is directly obtained by plugging Rényi divergences into the GWM and using Lemma~\ref{sym}. Let $\alpha > 1, z \geq 0$.
\begin{itemize}
    \item $\Lap(0,\rho I_d)$ is symmetric and  $z \mapsto g^{-1}(\sum_{i=1}^dg(z_i))$ is the $l_1$ norm  for $g : z \mapsto |z|$. For $L \sim \Lap(0,\rho I_d)$, \[D_\alpha(L+z,L) =  \frac{1}{\alpha-1}\log\left(\frac{\alpha}{2\alpha-1}e^{|z|(\alpha-1)/\rho} + \frac{\alpha-1}{2\alpha-1}e^{-|z|\alpha/ \rho} \right).\] Also, for $z \geq 0$:
    \[\frac{d}{dz^2}D_\alpha(L+z,L) =  \frac{\alpha (2\alpha^2-1)}{\rho^2 (2\alpha-1)}\frac{e^{-z/\rho}}{\alpha e^{z(\alpha-1)/\rho} + (\alpha-1)e^{-z\alpha/ \rho}} \geq 0.\]
    \item $\mathcal{N}(0,\sigma^2 I_d)$  is symmetric and  $z \mapsto g^{-1}(\sum_{i=1}^dg(z_i))$ is the $l_2$ norm for $g : z \mapsto z^2$. For $N \sim \mathcal{N}(0,\sigma^2 I_d)$, $D_\alpha(N+z,N) = \frac{\alpha z^2}{2\sigma^2}$.
\end{itemize}
\end{proof}

\subsection{Utility of the GWM (Proposition~\ref{GWMUtility})}
\label{appendixGWMUtility}

Below, we make the informal result of Proposition~\ref{GWMUtility} precise and provide its proof.

\begin{proposition}[Utility of the GWM]
    Let $n, d, d_1, \dots, d_n \in \mathbb{N}^*$, $\mathcal{X}\subset \mathbb{R}^{d}$. Let $(\mathcal{S},\mathcal{Q},\Theta)$ be a Pufferfish framework such that, for each $\theta \in \Theta$, $\theta = \otimes_{k=1}^{n}\theta_k$, with $\theta_k \in \mathcal{P}(\mathcal{X}^{d_k})$. We note $X = (X_1^1,\dots,X_{d_1}^1,\dots,X_{d_n}^{n})\sim \theta$. We assume that $s_{i,k}^a = \{X_i^k = a\} \in \mathcal{S}$ and $\mathcal{Q} = \{(s_{i,k}^a,s_{i,k}^b); k \in \{1,\dots,n\}, i \in \{1,\dots,d_k\}, a,b \in \mathcal{X}\}$. Following~\citet{Song2017}, we define the corresponding group differential privacy of the Pufferfish framework as: $G_k = (x_1^k,\dots,x_{d_k}^k) \in \mathcal{X}^{d_k}$ and $D_k = \{(x,x') \in \mathcal{X}^{d_k}\text{ such that $x$ and $x'$ only differ in $G_k$}\}$.\[\Delta_{GROUP}(f) = \max_{k \in \{1,\dots,n\}}\max_{(x,x')\in D_k}\|f(x)-f(x')\|.\]
    Then, $\Delta_G \leq \Delta_{GROUP}(f)$.
\end{proposition}
\begin{proof}
    Let $(s_{i,l}^a,s_{i,l}^b) \in \mathcal{Q}, \theta \in \Theta$, with $\theta = \otimes_{k=1}^{n}\theta_k$. Let $Y \sim P(f(X)|s_{i,l}^a,\theta)$. Let $Z \sim \theta_l|s_{i,l}^b$ drawn independently from $Y$. For $k\in \Iintv{1,n}, i \in \Iintv{1,d_k}$. We define $Y^{'k}_{~i} = \begin{cases}
        &Y_i^k \text{ if } k \neq l\\
        &Z_i \text{ else}
    \end{cases}$ and $Y' = (Y_1^{'1},\dots,Y_{d_1}^{'1},\dots,Y_{d_n}^{'n})$.
    
    Then, $(Y,Y') \in D_l$, $Y' \sim P(f(X)|s_{i,l}^b)$ and:
    \[\|Y-Y'\| \leq \max_{(x,x')\in D_l}\|f(x)-f(x')\| \leq \Delta_{GROUP}(f).\]
    Then, $W_\infty(P(f(X)|s_{i,l}^a),P(f(X)|s_{i,l}^b)) \leq \Delta_{GROUP}(f)$ and $\Delta_G \leq \Delta_{GROUP}(f)$.
\end{proof}

\section{Approximate General Wasserstein Mechanism (Section~\ref{sectionapproxwasserstein})}

Our result relies on the following characterization of $(z,\delta)$-proximity.
\begin{lemma}
\label{close}
    $\mu$ and $\nu$ are $(z,\delta)$-near iff $\; \exists X \sim \mu, Y \sim \nu$ and $V \in \mathcal{P}(\mathbb{R}^d)$ such that $X+V = Y$ and $P(\|V\| > z) < \delta$.
\end{lemma}
\begin{proof}
    Let $z \geq 0$, $\delta \in (0,1)$, $\mu$, $\nu$ two distributions on $\mathbb{R}^d$ such that $\mu$ and $\nu$ are $(z,\delta)$-near. Then, there exists $\pi$ a coupling between $\mu$ and $\nu$ such that $\int_\mathcal{R}d\pi(x,y) \geq 1- \delta$ and $\forall (x,y) \in \mathcal{R}, \|x-y\| \leq z$. We note $V = Y-W$ where $(W,Y)$ is drawn from the coupling $\pi$. We observe that $\mathcal{R} \subset \{(x,y) ; \|x-y\| \leq z\}$. \\
    
    Then, $P(\|V\| > z) \leq P((W,Y) \notin \mathcal{R}) = \int_{\mathbb{R}^d\setminus\mathcal{R}}d\pi(x,y) < \delta$.
    
    For the opposite side, consider the coupling $\pi$ of the pair $(W,Y)$ such that $W \sim \mu, Y \sim \nu$ and $W+V = Y$ with $P(\|V\| > z) < \delta$. 
    
    Then, $P(\|V\| \leq z) = \int_{\|x-y\| < z}d\pi(x,y) \geq 1-\delta$.
\end{proof}

\subsection{Proof of Lemma~\ref{shiftApprox}}
\begingroup
\def\thelemma{\ref{shiftApprox}}
\begin{lemma}[Approximate shift reduction]
Let $\mu, \nu, \zeta$ be three distributions on $\mathbb{R}^d$. We denote $D_\alpha^{(z, \delta)}(\mu,\nu) = \underset{\mu,\mu'~(z,\delta)\text{-near}}{\inf}D_\alpha(\mu',\nu)$. Then, for all $\delta \in (0,1)$, there exists an event $E$ such that $P(E) \geq 1-\delta$ and:  \[D_\alpha\left((\mu \ast \zeta)_{|E},(\nu \ast \zeta)\right) \leq D_\alpha^{(z,\delta)}(\mu,\nu) + R_\alpha(\zeta,z) + \frac{\alpha}{\alpha-1}\log\Big(\frac{1}{1-\delta}\Big).\]
\end{lemma}
\addtocounter{lemma}{-1}
\endgroup
\begin{proof}
    Let $\alpha >1$, $z > 0$, $X \sim \mu, Y \sim \nu, N \sim \zeta$ and $W \sim \xi \in \mathcal{P}(\mathbb{R}^d)$ such that $P(\|W\| \geq z) = \delta$ and $N$ is independent of $X$, $Y$ and $W$. We use the abuse of notation $D_\alpha(\mu,\nu) = D_\alpha(X,Y)$, with $X\sim \mu$, $Y\sim \nu$. We consider the event $E = \{\|W\| \leq z\}$. Like in the original proof of the shift reduction lemma of~\citet{Feldman2018}, we have:
    \[ D_\alpha((X+N)_{|E},Y+N) =  D_\alpha((X+W+N-W)_{|E},Y+N)  \leq D_\alpha((X+W,N-W)_{|E},(Y,N)).\] by post-processing (Proposition~\ref{postproc}) for $\mathcal{M}_1(x,y) = x+y$. Then, we have:
    \begin{align*}
        & D_\alpha((X+W,N-W)_{|E},(Y,N)) \\&= \frac{1}{\alpha-1}\log\left(\int\frac{P_{(X+W,N-W)_{|E}}(x,y)^\alpha}{P_{Y,N}(x,y)^{\alpha-1}}dxdy\right) \\
        &= \frac{1}{\alpha-1}\log\left(\int\frac{P_{X+W|E}(x)^\alpha P_{N-W|E,X+W=x}(y)^\alpha}{\nu(x)^{\alpha-1}\zeta(y)^{\alpha-1}}dxdy\right) \\
        &= \frac{1}{\alpha-1}\log\left(\int\frac{P_{X+W|E}(x)^\alpha}{\nu(x)^{\alpha-1}} \left(\int\frac{P_{N-W|E,X+W=x}(y)^\alpha}{\zeta(y)^{\alpha-1}}dy\right)dx\right) \\
        &= \frac{1}{\alpha-1}\log\int\frac{P_{X+W|E}(x)^\alpha}{\nu(x)^{\alpha-1}}\left(\int\frac{\left(\int_{\|u\| \leq z} P_{N-W|X+W=x,W=u}(y)\xi(u) du\right)^{\alpha}}{\zeta(y)^{\alpha-1}}dy\right)dx\\
        &\leq \frac{1}{\alpha-1}\log\left(\int\frac{P_{X+W|E}(x)^\alpha}{\nu(x)^{\alpha-1}} \left(\int_{\|u\| \leq z}\frac{\zeta(y+u)^\alpha}{\zeta(y)^{\alpha-1}}\xi(u)dudy\right)dx\right)\\
        &\leq \frac{1}{\alpha-1}\log\left(\int\frac{P_{X+W|E}(x)^\alpha} {\nu(x)^{\alpha-1}}dx\right) + R_\alpha(\zeta,z).  
    \end{align*}
    Yet, \[P_{X+W|E}(x)^\alpha = \left(\frac{P_{X+W}(x) - P(\bar{E})P_{X+W|\bar{E}}(x)}{P(E)}\right)^\alpha \leq \frac{P_{X+W}(x)^\alpha}{(1-\delta)^\alpha}.\]
    Thus: \begin{align*}
        &D_\alpha((X+W,N-W)_{|E},(Y,N)) \\&\leq D_\alpha(X+W,Y) + R_\alpha(\zeta,z) - \frac{\alpha}{\alpha-1}\log\left(1-\delta\right).\qedhere
    \end{align*}
\end{proof}
\subsection{Relationship between $(+\infty,\varepsilon,\delta)$-approximate RPP and $(\varepsilon,\delta)$-PP}

\begin{proposition}
\label{approxrpptoinftypp}
    If $\mathcal{M}$ is $(+\infty,\varepsilon,\delta)$-approximate RPP, then it is $(\varepsilon,\delta)$-PP.
\end{proposition}

\begin{proof}
    The proof uses the same approach as the Lemma 8.8 of \citet{Bun2016}. Let $\left(s_{i}, s_{j}\right) \in \mathcal{Q}$, $\theta \in$ $\Theta$. Without loss of generality, we assume that there exists $E, E'$ such that $P(E) = 1- \delta, P(E') = 1 - \delta$ and we have: $D_\infty\left(P\left(\mathcal{M}(X)=w \mid s_{i}, \theta, E\right), P\left(\mathcal{M}(X)=w \mid s_{j}, \theta, E'\right)\right) \leq \varepsilon$. Then,
    \[\underset{w \in Range(\mathcal{M})}
    {\sup}\log\frac{P\left(\mathcal{M}(X)=w \mid s_{i}, \theta, E\right)}{ P\left(\mathcal{M}(X)=w \mid s_{j}, \theta, E'\right)}\leq \varepsilon.\]
    \begin{align*}
    P\left(\mathcal{M}(X)=w \mid s_{j}, \theta\right)
    &= P(E')P\left(\mathcal{M}(X)=w \mid s_{j}, \theta,E'\right) + P(\bar{E'})P\left(\mathcal{M}(X)=w \mid s_{j}, \theta,\bar{E'}\right)\\
    &\geq (1-\delta)P\left(\mathcal{M}(X)=w \mid s_{j}, \theta,E'\right),\\
    P\left(\mathcal{M}(X)=w \mid s_{i}, \theta\right) &= P(E)P\left(\mathcal{M}(X)=w \mid s_{i}, \theta,E\right) + P(\bar{E})P\left(\mathcal{M}(X)=w \mid s_{i}, \theta,\bar{E}\right)\\
    &\leq (1-\delta)P\left(\mathcal{M}(X)=w \mid s_{i}, \theta,E\right)+\delta\\
    &\leq (1-\delta) P\left(\mathcal{M}(X)=w \mid s_{j}, \theta,E'\right)e^\varepsilon +\delta \\
    &\leq P\left(\mathcal{M}(X)=w \mid s_{j}, \theta\right)e^\varepsilon +\delta.\qedhere
    \end{align*}
\end{proof}

\begingroup
\def\theproposition{\ref{aproxRPPtoRPP}}
\begin{proposition}
     If $\mathcal{M}$ is $(\alpha,\varepsilon,\delta)$-approximate RPP, then it is $(\varepsilon',2\delta)$-PP, with $\varepsilon' = \varepsilon + \frac{\log(1/\delta)}{\alpha-1}$.
 \end{proposition}
 \addtocounter{proposition}{-1}
\endgroup

\begin{proof}
    We use the proof techniques of Proposition~\ref{approxrpptoinftypp} and Proposition~\ref{RPPtoPP}. Let $\varepsilon \geq 0, \alpha >1$. Let $(\mathcal{S},\mathcal{Q},\Theta)$ be a Pufferfish privacy framework and $\mathcal{M}$ an $(\alpha,\varepsilon,\delta)$-RPP mechanism. Let $\delta \in (0,1)$, $\theta \in \Theta$, $(s_i,s_j) \in \mathcal{Q}$ and $z \in Range(\mathcal{M)}$. There exists $E,E'$ such that $D_\alpha\left(P\left(\mathcal{M}(X) \mid s_{i}, \theta, E\right), P\left(\mathcal{M}(X) \mid s_{j}, \theta, E'\right)\right) \leq \varepsilon$ and $P(E),P(E') \geq 1-\delta$, The proof technique of Proposition~\ref{RPPtoPP} allows to show that:
    \[P\left(\mathcal{M}(X)=z|E,s_i,\theta\right) \leq e^{\varepsilon+\frac{\log(1/\delta)}{\alpha-1}} P\left(\mathcal{M}(X)=z|E',s_j,\theta\right) + \delta.\] Then, the proof technique of Proposition~\ref{approxrpptoinftypp} allows to show that: \[P\left(\mathcal{M}(X)=z|s_i,\theta\right) \leq e^{\varepsilon+\frac{\log(1/\delta)}{\alpha-1}} P\left(\mathcal{M}(X)=z|,s_j,\theta\right) + 2\delta.\]

\end{proof}
\subsection{Proof of Theorem~\ref{GAWM}}
\begingroup
\def\thetheorem{\ref{GAWM}}
\begin{theorem}[General approximate Wasserstein mechanism]
Let $f : \mathcal{D} \to \mathbb{R}^d$ be a numerical query. 
 For all $\delta \in (0,1)$, let us denote:
\begin{align*}
    &\Delta_{G,\delta} >\inf\{ z \in \mathbb{R}; \forall (s_i,s_j) \in S, \forall \theta \in \Theta, \\&\left( P((f(X)|s_i,\theta),P(f(X)|s_j,\theta)\right) \text{ are } (z,\delta) \text{-near}\}.
\end{align*}
Let $N = (N_1,\dots,N_d) \sim \zeta$, where $N_1,\dots,N_d$ are iid real random variables independent of the dataset $X$.
Then, $\mathcal{M} = f(X) + N$ satisfies $(\alpha,R_\alpha(\zeta,\Delta_{G,\delta})+ \frac{\alpha}{\alpha-1}\log\frac{1}{1-\delta},\delta)$-approximate RPP for all $\alpha \in (1,+\infty)$ and $(R_\infty(\zeta,\Delta_{G,\delta}) + \log\frac{1}{1-\delta},\delta)$-PP. 
\end{theorem}
\addtocounter{theorem}{-1}
\endgroup
\begin{proof}
    This proof is similar to Theorem~\ref{GWM} but we use the approximate shift reduction lemma (Lemma~\ref{shiftApprox}). We use the abuse of notation $D_\alpha(X_{|E},Y_{|E}) = D_\alpha(P(X|E),P(Y|E))$. Let $f : \mathcal{D} \to \mathbb{R}^d$ be a numerical query and $N = (N_1,\dots,N_d) \sim \zeta$, where $N_1,\dots,N_d$ are iid real random variables independent of the data $X$. 
 Let $\delta \in (0,1)$. Let us denote: \[\Delta_{G,\delta} >\inf\{ z \in \mathbb{R}; \forall (s_i,s_j) \in S, \forall \theta \in \Theta,  \left( P(f(X)|s_i,\theta),P(f(X)|s_j,\theta)\right) \text{ are } (z,\delta) \text{-near}\}.\]
 By the approximate shift reduction lemma (Lemma~\ref{shiftApprox}), there exists $E$ such that $P(E) \geq 1-\delta$ and: \[D_\alpha\left(\left(f(X) + N \right)_{\mid E, s_i, \theta},\left(f(X) + N \right)_{\mid s_j, \theta}\right) \leq D_\alpha^{(z,\delta)}\left(f(X) _{\mid s_i, \theta},f(X)_{\mid s_j, \theta}\right) + R_\alpha(\zeta,z) - \frac{\alpha}{\alpha-1}\log(1-\delta).\]
 By definition, \[D_\alpha^{(z,\delta)}\left(f(X) _{\mid s_i, \theta},f(X)_{\mid s_j, \theta}\right) = \underset{\mu \in \mathcal{P}(\mathbb{R}^d); \mu,P(f(X)\mid s_i, \theta) \text{ are } (z,\delta)\text{-near}}{\inf}D_\alpha\left(\mu,P(f(X)\mid s_j, \theta\right)),\]
and \[D_\alpha^{(\Delta_{G,\delta},\delta)}\left(f(X) _{\mid s_i, \theta},f(X)_{\mid s_j, \theta}\right) = 0.\] Then, \[D_\alpha \left(\left(f(X) + N \right)_{\mid E, s_i, \theta},\left(f(X) + N \right)_{\mid s_j, \theta}\right) \leq R_\alpha(\zeta,\Delta_{G,\delta}) - \frac{\alpha}{\alpha-1}\log(1-\delta).\qedhere\]

\end{proof}
\subsection{Result for Usual Noise Distributions}
\label{completeapproxwasserstein}

We provide below a corollary of Theorem~\ref{GAWM} that gives closed formula for usual noise distributions to get approximate RPP guarantees. 

\begin{proposition}[Approximate Wasserstein mechanism]
We note $I_d$ the identity matrix of size $d$. The results are similar to those of the general Wasserstein mechanism (Corollary~\ref{wassersteinusualdistrib}), but with an additive term which depends on $\delta$:
\begin{itemize}
    \item $\mathcal{M}(X) = X + N$ with $N \sim \mathcal{N}\Big(0,\frac{\alpha \Delta_{G,\delta}^2 }{2\left(\varepsilon + \frac{\alpha}{\alpha-1}\log(1-\delta)\right)}I_d\Big)$ is $(\alpha, \varepsilon, \delta)$-approximate RPP.
    \item $\mathcal{M}(X)  = X + L$ with $L \sim \Lap(0,\rho I_d)$ is $\big(\alpha,\frac{1}{\alpha-1}\left(\log\left(b\right)-\alpha\log(1-\delta) \right), \delta\big)$-approximate RPP for $b = \frac{\alpha}{2\alpha-1}e^{\Delta_{G,\delta}(\alpha-1) /\rho} + \frac{\alpha-1}{2\alpha-1}e^{-\Delta_{G,\delta} \alpha / \rho}$.
    \item  $\mathcal{M}(X)  = X + L$ with $L \sim \Lap\big(0,\frac{\Delta_{G,\delta}}{\varepsilon +\log(1-\delta)} I_d\big)$ is $(\varepsilon, \delta)$-PP.
\end{itemize}
\end{proposition}

\subsection{Relationship with Distribution Privacy Results of~\citet{Chen2023}}
\label{proofDistrPP}

We start by recalling the definition of distribution privacy.

\begin{definition}[Distribution privacy~\cite{Chen2023}]
    A mechanism $\mathcal{M}$ satisfies $(\varepsilon, \delta)$-distribution privacy with respect to a set of distribution pairs $\Psi \subset \Theta \times \Theta$ if for all pairs $(\psi_i,\psi_j) \in \Psi$ and all subsets $S \subset \text{Range}(\mathcal{M})$,
    \[P(\mathcal{M}(X) \in S | \psi_i) \leq e^\varepsilon P(\mathcal{M}(X) \in S | \psi_j) + \delta,\]
    where the expression $P(\mathcal{M}(X) \in S | \psi)$ denotes the probability that $\mathcal{M}(X$) given $X \sim \psi$.
\end{definition}

For completeness, we recall the original approximate Wasserstein mechanism Theorem for distribution privacy from~\cite{Chen2023}.

\begin{theorem}[Approximate Wasserstein mechanism for distribution privacy~\cite{Chen2023}]
    Let $(\Psi,\Theta)$ be a distribution privacy framework. Let $W > 0$, $\delta \in (0,1)$. Suppose that for all $(\psi_i,\psi_j) \in \Psi$, $P(X\mid\psi_i)$ and $P(X\mid\psi_j)$ are $(W,\delta)$-near. Then $\mathcal{M}(X) = X + L$ where $L \sim \text{Lap}\big(0,\frac{W}{\varepsilon}I\big)$ is $(\varepsilon,\delta)$-distribution private.
\end{theorem}
We now formally state and prove the equivalence between Pufferfish privacy and distribution privacy.

\begin{proposition}
     Let $(E, \mathcal{B}(E))$ be a measurable space, where $|E| \leq \aleph_1$ is a topological space with its Borel $\sigma$-algebra $\mathcal{B}(E)$ and $\aleph_1$ is the cardinality of $\mathbb{R}$. Let $\Theta \subset \mathcal{P}(\mathcal{B}(E))$. Let $(\mathcal{S},\mathcal{Q},\Theta)$ be a Pufferfish privacy instance and $\mathcal{M}$ a randomized mechanism. Then, there exists a distribution privacy instance $(\Psi, \Theta')$ such that $\mathcal{M}$ is $(\varepsilon,\delta)$-PP iff $\mathcal{M}$ is $(\varepsilon,\delta)$-distribution private. Conversely, let $(\Psi, \Theta)$ be a distribution privacy instance. Then, there exists a Pufferfish privacy instance $(\mathcal{S},\mathcal{Q},\Theta')$ such that $\mathcal{M}$ is $(\varepsilon,\delta)$-PP iff $\mathcal{M}$ is $(\varepsilon,\delta)$-distribution private.
\end{proposition}
\begin{remark}
    The condition $|E| \leq \aleph_1$ is quite general. In particular, it allows the data space to be (a subset of) $\mathbb{R}^d$, thus covering typical data domains found in fields like data analysis, machine learning, text processing, computer vision, and database management. 
\end{remark}
\begin{proof}
We show the equivalence between the Pufferfish privacy framework and the distribution privacy framework. Let $\Theta \subset \mathcal{P}(\mathcal{B}(E))$, where $|E| \leq \aleph_1$.
    \begin{itemize}
        \item Let $(\mathcal{S},\mathcal{Q}, \Theta)$ be a Pufferfish privacy instance. We consider:  
        \[\Psi = \left\{ \left(P(X|s_i,\theta),P(X|s_j,\theta)\right) \text{ such that } (s_i,s_j) \in \mathcal{Q}, \theta \in \Theta \text { and } P\left(s_{i} \mid \theta\right) \neq 0, P\left(s_{j} \mid \theta \right)\neq 0 \right\}.\]Then,
    \begin{align*}
        &\forall w \in \text{Range}(\mathcal{M)}, \forall (\psi_i,\psi_j) \in \Psi ,\\ 
        &P\left(\mathcal{M}(X)=w \mid \psi_i\right)\leq e^{\varepsilon}P\left(\mathcal{M}(X)=w \mid \psi_j\right) + \delta\\
        &\iff\\
        &\forall w \in \text{Range}(\mathcal{M)} ,\forall (s_i,s_j) \in \mathcal{Q}, \theta \in \Theta \text{ such that } P\left(s_{i} \mid \theta\right) \neq 0, P\left(s_{j} \mid \theta \right)\neq 0,\\
        &P\left(\mathcal{M}(X)=w \mid s_{i}, \theta\right)\leq e^{\varepsilon}P\left(\mathcal{M}(X)=w \mid s_{j}, \theta\right) + \delta.
    \end{align*}
    \item Let $(\Psi,\Theta)$ be a distribution privacy instance. First, we consider the case where each $\psi \in \Theta$ is parametrized by a vector $\rho \in \mathbb{R}^d$, which means that there exists a bijection between a subset of $\mathbb{R}^d$ and $\Theta$. For $\rho \in \mathbb{R}^d$, if it exists, we denote $\psi_\rho \in \Theta$ the corresponding distribution. Then, we denote $\Phi = \{\rho \in \mathbb{R}^d \text{ such that } \exists \psi \in \Psi; (\psi_\rho,\psi) \in \Psi \vee (\psi,\psi_\rho) \in \Psi\}$ and $\Omega = \{(\rho_1,\rho_2) \in \Phi \times \Phi \text{ such that } (\psi_{\rho_1},\psi_{\rho_2)} \in \Psi\} \subset \mathbb{R}^{n\times 2}$ and $\Pi = \{\pi \in P(\mathcal{B}(\mathbb{R}^d)) \text{ such that } supp(\pi) = \Phi\}$. We consider:
    \begin{align*}
            \mathcal{S} &= \{(s_\rho = \text{``}X \text{ has been generated from the distribution } \psi_{\rho}\text{''}), \forall \rho \in \Phi\}, \\
            \mathcal{Q} &= \{(s_{\rho_1},s_{\rho_2}) \text{ such that } (\rho_1,\rho_2) \in \Omega\}, \\ \Theta' &= \left\{\theta_\pi \in \mathcal{P}(\mathcal{B}(E)) \text{ such that } \pi \in \Pi \wedge P(X|\theta_\pi)  = \int_{\Phi}\pi(\rho)P(X|\psi_\rho)d\rho\right\}.
        \end{align*}
    Then, $\forall w \in \text{Range}(\mathcal{M)} ,\forall (s_i,s_j) \in \mathcal{Q}, \theta_\pi \in \Theta'$, $P\left(X|\theta_\pi,s_i\right)=P\left(X|\psi_i\right)$. Thus, we have:
    \begin{align*}
        &\forall w \in \text{Range}(\mathcal{M)}, \forall (\psi_i,\psi_j) \in \Psi ,\\ 
        &P\left(\mathcal{M}(X)=w \mid \psi_i\right)\leq e^{\varepsilon}P\left(\mathcal{M}(X)=w \mid \psi_j\right) + \delta\\
        &\iff\\
        &\forall w \in \text{Range}(\mathcal{M)} ,\forall (s_i,s_j) \in \mathcal{Q}, \theta \in \Theta \text{ such that } P\left(s_{i} \mid \theta\right) \neq 0, P\left(s_{j} \mid \theta \right)\neq 0,\\
        &P\left(\mathcal{M}(X)=w \mid s_{i}, \theta\right)\leq e^{\varepsilon}P\left(\mathcal{M}(X)=w \mid s_{j}, \theta\right) + \delta.
    \end{align*}
    In this proof, the case $|\Theta| = n \in \mathbb{N}^*$ is a case where $\psi \in \Theta$ can be parameterized. One such parameterization is to define $\Theta = \{\psi_1,\dots,\psi_n\}$ and the mapping $i \in \mathbb{N} \mapsto \psi_i \in \Theta$.
    
    The second part of the proof relies on the fact that the distributions of $\mathcal{P}(\mathcal{B}(E))$ are parameterizable. The hypothesis $|E| \leq \aleph_1$ allows us to reduce to the case $E = \mathbb{R}$, up to a bijection. Yet, every distribution of $\mathcal{B}(\mathbb{R})$ is entirely defined by its values taken on open intervals of $\mathbb{R}$ and each open interval of $\mathbb{R}$ is a countable union of open intervals with rational endpoints. Therefore, $|\mathcal{P}(\mathcal{B}(\mathbb{R}))| \leq 2^{\aleph_0} = \aleph_1$, where the notation $\aleph_0$ denotes the cardinal of $\mathbb{N}$ and we can map every distribution of $\mathbb{R}$ with elements of $\mathbb{R}$.    
    \end{itemize}
\end{proof}
\begin{remark}
    The proof shows how to transition from the Pufferfish privacy framework to the distribution privacy framework. Thus, it is possible to use Pufferfish private mechanisms to achieve distribution privacy guarantees (and vice versa).
\end{remark}
This equivalence result allows us to precisely compare our result (Theorem~\ref{GAWM}) to the result of~\citet{Chen2023}.
Our approximate shift reduction result (Lemma~\ref{shiftApprox}) induces an additional term which prevents us from recovering exactly the results of~\citet{Chen2023} in the particular case of the Laplace mechanism for PP. However, we believe that our analysis can be improved and lead to better results. More generally, our result can be used with a wide range of noise distributions and in the RPP framework, which is more general than PP (and thus more general than distribution privacy).

\subsection{Utility of the GAWM (Proposition~\ref{GAWMUtility})}
\label{appendixGAWMUtility}

Below, we make the informal result of Proposition~\ref{GAWMUtility} precise and provide its proof.

\begin{proposition}[Utility of the GAWM]
    Let $(\mathcal{S},\mathcal{Q},\Theta)$ be a Pufferfish framework, $\delta \in (0,1), \alpha > 1$ and let $\mathcal{M}(X) = f(X) + N$, where $X \sim \theta \in \Theta$, $N \sim \zeta$ and $f$ is a numerical query. Then, $\Delta_G$ as defined in Theorem~\ref{GWM} is greater or equal than $\Delta_{G,\delta}$ defined in Theorem~\ref{GAWM}. Moreover, if $R_\alpha(\Delta_{G,\delta},\zeta) \leq R_\alpha(\Delta_G,\zeta)+\frac{\alpha}{\alpha-1}\log(1-\delta)$ then the GAWM achieves better utility than the GWM with $(\alpha,\varepsilon,\delta)$-RPP, without additional privacy cost on the $\varepsilon$. It happens when $\Delta_G$ is sufficiently larger than $\Delta_{G,\delta}$, which happens when there exists $(s_i,s_j) \in \mathcal{Q}, \theta \in \Theta$ and $(Y,Y') \sim \pi \in \Gamma(P(f(X)|s_i,\theta),P(f(X)|s_j,\theta))$ such that $\|Y-Y'\|$ is large with small probability.
\end{proposition}

\begin{proof}
    Let $(s_i,s_j) \in \mathcal{Q}, \theta \in \Theta$. Then, there exists $(Y,Y') \sim \pi \in \Gamma(P(f(X)|s_i,\theta),P(f(X)|s_j,\theta))$ such that $P(\|Y-Y'\| > \Delta_{G}) =0$. Then, for any $\delta \in (0,1)$, $P(\|Y-Y'\| > \Delta_{G}) < \delta$  and by Lemma~\ref{close}, $Y$ and $Y'$ are $(\Delta_{G},\delta)$-near. Finally, $\Delta_{G,\delta} \leq \Delta_G$.
\end{proof}
\section{Leveraging $W_p$ metrics (Section~\ref{sectionrelaxWasserstein})}
\subsection{Proof of Lemma~\ref{generalizedShift}}
\begingroup
\def\thelemma{\ref{generalizedShift}}
\begin{lemma}[Generalized shift reduction]
    Let $\zeta$ be a noise distribution of $\mathbb{R}^d$. Let $z, p ,q > 0$ such that $1/p+1/q=1$.     We note :
    \[D_{\alpha, \alpha', \zeta}^{(z)}(\mu,\nu) = \underset{\xi; \underset{W \sim \xi}{\mathop{\mathbb{E}}}[\exp((\alpha'-1)D_{\alpha'}(\zeta ,\zeta \ast W))]\leq z}{\inf}D_\alpha(\mu \ast \xi , \nu).\]
    Then, we have : \[D_\alpha(\mu \ast \zeta,\nu \ast \zeta) \leq D_{p(\alpha-1)+1,q(\alpha-1)+1, \zeta}^{(z)}(\mu,\nu) + \frac{\log(z)}{q(\alpha-1)}.\]
    In the case $q = 1$: \[D_\alpha(\mu \ast \zeta,\nu \ast \zeta) \leq D_{\infty,\alpha, \zeta}^{(z)}(\mu,\nu) + \frac{\log(z)}{\alpha-1}.\]
\end{lemma}
\addtocounter{lemma}{-1}
\endgroup
\begin{proof}
    The proof construction is similar to the one developed in~\cite{Chen2023}. We do not apply Jensen inequality at the last step of the proof to obtain Orlicz-Wasserstein metrics, and keep the result general and working for a broader range of distributions. We use the abuse of notation $D_\alpha(\mu,\nu) = D_\alpha(X,Y)$, with $X\sim \mu$, $Y\sim \nu$.
    Let $z > 0$, $X \sim \mu, Y \sim \nu, N \sim \zeta \in \mathcal{P}(\mathbb{R}^d)$ be a noise distribution and $W \sim \xi \in \mathcal{P}(\mathbb{R}^d)$ such that: \[\mathbb{E}_W[\exp(q(\alpha-1)D_{q(\alpha-1)+1}(\zeta,\zeta\ast W))] \leq z.\] Let $p,q > 0$ such that $\frac{1}{p} + \frac{1}{q}$ = 1. We want to compute : $D_\alpha(X+N,Y+N)$. By the post processing theorem applied on the map $f : (x,y) \to x+y $, and the fact that $X+N = X+W-W+N$, we have : \[D_\alpha(X+N,Y+N) \leq D_\alpha((X+W,N-W),(Y,N)).\]
    We have:
    \begin{align*}
        D_\alpha((X+W,N-W),(Y,N)) &= \frac{1}{\alpha-1}\log\left(\int\frac{P_{(X+W,N-W)}(x,y)^\alpha}{P_{Y,N}(x,y)^{\alpha-1}}dxdy\right) \\
        &= \frac{1}{\alpha-1}\log\left(\int\frac{P_{X+W}(x)^\alpha P_{N-W|X+W = x}(y)^{\alpha-1}}{\nu(x)^{\alpha-1}\zeta(y)^{\alpha}}dxdy\right)\\
        & = \frac{1}{\alpha-1}\log \mathbb{E}_{\substack{U \sim X+W \\ V \sim N-W|X+W = U}}\left[\left(\frac{P_{X+W}(U)}{\nu(U)}\right)^{\alpha-1}\left(\frac{P_{N-W|X+W = x}(V)}{\zeta(V)}\right)^{\alpha-1}\right] \\
        & \leq \frac{1}{p(\alpha-1)}\log \mathbb{E}_{U \sim X+W}\left[\left(\frac{P_{X+W}(U)}{\nu(U)}\right)^{p(\alpha-1)}\right] (1)\\&+ \frac{1}{q(\alpha-1)}\log\mathbb{E}_{\substack{U \sim X+W \\ V \sim N-W|X+W = U}}\left[\left(\frac{P_{N-W|X+W = x}(V)}{\zeta(V)}\right)^{q(\alpha-1)}\right] (2) \substack{\text{ by Hölder} \\ \text{ inequality}}
    \end{align*}
Immediately $(1) = D_{p(\alpha-1)+1}(X+W,Y)$ and, given that
\begin{align*}
    P_{N-W|X+W = x}(y)^{q(\alpha-1)+1} &= \left(\int P_{N-W|W=z}(y)\xi(z)dz\right)^{q(\alpha-1)+1} \\ &= \mathbb{E}_W\left[\zeta(y+W)\right]^{q(\alpha-1)+1} \\
    &\leq \mathbb{E}_W\left[\zeta(y+W)^{q(\alpha-1)+1}\right],
\end{align*}
we have:
\begin{align*}
    (2) &= \frac{1}{q(\alpha-1)}\log \int \left(\frac{P_{N-W|X+W = x}(y)}{\zeta(y)}\right)^{q(\alpha-1)}P_{X+W}(x)P_{N-W|X+W = x}(y)dxdy \\
   &\leq \frac{1}{q(\alpha-1)}\log \int \frac{\zeta(y+u)^{q(\alpha-1)+1}}{\zeta(y)^{q(\alpha-1)}}\xi(u)P_{X+W}(x)dxdydu\\
   &\leq \frac{1}{q(\alpha-1)} \log \mathbb{E}_{W\sim \xi}\left[\exp(q(\alpha-1)D_{q(\alpha-1)+1}(\zeta,\zeta \ast W))\right]\\
   &\leq \frac{\log(z)}{q(\alpha-1)}.
\end{align*}
In the case $p = +\infty$, let $W \sim \xi \in \mathcal{P}(\mathbb{R}^d)$ such that: \[\mathbb{E}_W[\exp((\alpha-1)D_{\alpha}(\zeta,\zeta \ast W))] \leq z.\]
\begin{align*}
        D_\alpha((X+W,N-W),(Y,N)) &\leq \sup_{U \sim X+W}\frac{1}{\alpha-1}\log \left(\frac{P_{X+W}(U)}{\nu(U)}\right)^{\alpha-1}(3)\\&+ \frac{1}{(\alpha-1)}\log\mathbb{E}_{\substack{U \sim X+W \\ V \sim N-W|X+W = U}}\left[\left(\frac{P_{N-W|X+W = x}(V)}{\zeta(V)}\right)^{\alpha-1}\right] (4)\\
    \end{align*}
    Yet, $(3) =  D_\infty(P_{X+W},\nu)$ and:
    \begin{align*}
    (4) &= \frac{1}{\alpha-1}\log \int \left(\frac{P_{N-W|X+W = x}(y)}{\zeta(y)}\right)^{\alpha-1}P_{X+W}(x)P_{N-W|X+W = x}(y)dxdy \\
   &\leq \frac{1}{\alpha-1} \log \mathbb{E}_{W\sim \xi}\left[\exp((\alpha-1)D_{\alpha}(\zeta,\zeta \ast W))\right]\\
   &\leq \frac{\log(z)}{\alpha-1}.\qedhere
\end{align*}
\end{proof}

\subsection{Proof of Theorem~\ref{EGWM}}
\begingroup
\def\thetheorem{\ref{EGWM}}
\begin{theorem}[Distribution Aware General Wasserstein Mechanism]
Let $f : \mathcal{D} \to \mathbb{R}^d$ be a numerical query and  $\zeta$ a noise distribution of $\mathbb{R}^d$. Let $q \geq 1$. For $(s_i,s_j) \in \mathcal{Q}, \theta \in \Theta$, we note $\mu_i^\theta = P(f(X)|s_i,\theta)$. We denote:\[\Delta^{\zeta,q,\alpha}_G = \underset{\substack{(s_i,s_j) \in S\\ \theta_i\in \Theta}}{\max}\text{ }\underset{P(X,Y) \in \Gamma(\mu_i^\theta,\mu_j^\theta)}{\inf}\mathbb{E}\left[e^{q(\alpha-1)D_{q(\alpha-1)+1}(\zeta ,\zeta \ast (X-Y) )}\right].\]

Let $N = (N_1,\dots,N_d) \sim \zeta$, where $N_1,\dots,N_d$ are iid real random variables independent of the data $X$. Then, $\mathcal{M}(X) = f(X) + N$ satisfies $(\alpha,\frac{\log(\Delta_G^{\zeta,q,\alpha})}{q(\alpha-1)})$-RPP for all $\alpha \in (1,+\infty)$ and $\lim_{\alpha \to + \infty}\frac{\log(\Delta_G^{\zeta,q,\alpha})}{q(\alpha-1)}$-PP. 
\end{theorem}
\addtocounter{theorem}{-1}
\endgroup
\begin{proof}
The proof is similar to Theorem~\ref{GWM} but we use the generalized shift reduction lemma (Lemma~\ref{generalizedShift}). Let $(\mathcal{S},\mathcal{Q},\Theta)$ be a Pufferfish privacy instance. Let $f : \mathcal{D} \to \mathbb{R}^d$ be a numerical query and denote:

\[\Delta^{\zeta,q,\alpha}_G = \underset{\substack{(s_i,s_j) \in S\\ \theta_i\in \Theta}}{\max}\text{ }\underset{P(X,Y) \in \Gamma(\mu_i^\theta,\mu_j^\theta)}{\inf}\mathbb{E}\left[e^{q(\alpha-1)D_{q(\alpha-1)+1}(\zeta ,\zeta \ast (X-Y ))}\right].\]

Let $N = (N_1,\dots,N_d) \sim \zeta$, where $N_1,\dots,N_d$ are iid real random variables independent of the data $X$. Let $\alpha > 1$, $z > 0$, $(s_i,s_j)\in \mathcal{Q}$ and $\theta \in \Theta$. We use the abuse of notation $D_\alpha(X_{|E},Y_{|E}) = D_\alpha(P(X|E),P(Y|E))$. By the shift reduction lemma (Lemma~\ref{generalizedShift}), we have:
\[D_\alpha\left(\left(f(X) + N \right)_{\mid s_i, \theta},\left(f(X) + N \right)_{\mid s_j, \theta}\right) \leq  D_{p(\alpha-1)+1,q(\alpha-1)+1, \zeta}^{(z)}\left(f(X) _{\mid s_i, \theta},f(X)_{\mid s_j, \theta}\right) + \frac{\log(z)}{q(\alpha-1)}.\]
By definition, \[D_{p(\alpha-1)+1,q(\alpha-1)+1, \zeta}^{(z)}\left(f(X) _{\mid s_i, \theta},f(X)_{\mid s_j, \theta}\right) = \underset{W\in \mathcal{P}(\mathbb{R}^d);\underset{W \sim \xi}{\mathop{\mathbb{E}}}[\exp(q(\alpha-1)D_{q(\alpha-1)+1}(\zeta,\zeta \ast (W-f(X)_{\mid s_i, \theta})))]\leq z}{\inf} D_\alpha\left(W,f(X)_{\mid s_j, \theta}\right),\] and \[D_{p(\alpha-1)+1,q(\alpha-1)+1, \zeta}^{(\exp(q(\alpha-1)D_{q(\alpha-1)+1}(\zeta,\zeta \ast (f(X)_{\mid s_i, \theta}-f(X)_{\mid s_j, \theta}))))}\left(f(X) _{\mid s_i, \theta},f(X)_{\mid s_j, \theta}\right) = 0.\] Then, \[D_\alpha\left(\left(f(X) + N \right)_{\mid s_i, \theta},\left(f(X) + N \right)_{\mid s_j, \theta}\right) \leq  \frac{\log(\Delta^{\zeta,q,\alpha}_G)}{q(\alpha-1)}.\qedhere\]    
\end{proof}

\subsection{Proof of Corollary~\ref{cauchyMechanism}}

Divergences of shifts in Cauchy distributions have been discussed in~\cite{Verdu2023}. We generalize their results for certain types of generalized Cauchy distributions in the following lemma.
\begin{lemma}[Shifts of generalized Cauchy distributions]
\label{ShiftCauchy}
Let $k \in \mathbb{N}^*, \alpha > 1, \lambda > 0$ and $\beta_{k,\lambda} > 0$ such that $\zeta_{k,\lambda} : x \mapsto \beta_{k,\lambda}(\frac{1}{1+(\lambda x)^2})^{\frac{k}{2}}$ verifies $\int \zeta_{k,\lambda}(x) dx = 1$. Let $X \sim \zeta_{k,\lambda}$ and $z \geq 0$. Then,
\[D_\alpha(X+z,X) \leq \frac{1}{\alpha-1}\log \frac{\beta_{k,\lambda} \pi}{\lambda}Q_{k(\alpha-1)/2}\left(1+\frac{z^2}{\lambda^2}\right),\] where $Q_{k(\alpha-1)/2}$ is the Legendre function of the first kind of index $k(\alpha-1)/2$.
\end{lemma}

\begin{proof}
     Let $\lambda, z >0, k\in \mathbb{N}^*$. We have:
    \begin{align*}
        \int \frac{\zeta_{k,\lambda}(x+z)^\alpha}{\zeta_{k,\lambda}(x)^{\alpha-1}} dx &= \beta_{k,\lambda} \int \frac{\left(1+(\lambda (x-z))^2\right)^{(\alpha-1) k/2}}{\left(1+(\lambda x)^2\right)^{\alpha k/2}} dx \\
        &= \frac{\beta_{k,\lambda}}{\lambda} \int \frac{\left(1+(u-\lambda z)^2\right)^{(\alpha-1) k/2}}{\left(1+u^2\right)^{\alpha k/2}} du \\
        &= \frac{\beta_{k,\lambda}}{\lambda} \int_{-\pi /2}^{\pi/2} \frac{\left(1+(\tan(t)-\lambda z)^2\right)^{(\alpha-1) k/2}}{\left(1+\tan(t)^2\right)^{\alpha k/2}} (1+\tan^2(t)) dt\\
        &= \frac{\beta_{k,\lambda}}{\lambda} \int_{-\pi /2}^{\pi/2} \left(1+\tan^2(t)-2 \tan(t)\lambda z + \lambda^2 z^2\right)^{(\alpha-1) k/2} (\cos^2(t))^{\alpha k/2 - 1} dt\\
        &= \frac{\beta_{k,\lambda}}{\lambda} \int_{-\pi /2}^{\pi/2} \left(\cos^2(t)(1+\tan^2(t)-2 \tan(t)\lambda z + \lambda^2 z^2)\right)^{(\alpha-1) k/2} (\cos^2(t))^{k/2 - 1} dt\\
        &\leq \frac{\beta_{k,\lambda}}{\lambda} \int_{-\pi /2}^{\pi/2} \left(1-2 \sin(t)\cos(t)\lambda z + \cos^2(t) \lambda^2 z^2)\right)^{(\alpha-1) k/2}dt\\
        &\leq \frac{\beta_{k,\lambda}}{2\lambda} \int_{-\pi}^{\pi} \left(1-\sin(t)\lambda z + 
        (\cos(t)+1) \lambda^2 z^2/2)\right)^{(\alpha-1) k/2}dt\\
        &\leq \frac{\beta_{k,\lambda}}{2\lambda} \int_{-\pi}^{\pi} \left(1 + \lambda^2 z^2/2 -\sin(t)\lambda z + 
        \cos(t) \lambda^2 z^2/2)\right)^{(\alpha-1) k/2}dt\\
        &\leq \frac{\beta_{k,\lambda}}{2\lambda} \int_{-\pi}^{\pi} \left(1 + \lambda^2 z^2/2 +\sqrt{\lambda z + \lambda^2 z^2/2}\cos(t)\right)^{(\alpha-1) k/2}dt,\\
    \end{align*}
    And $Q_\alpha(z)$ is defined by: \[Q_\alpha(z) = \frac{1}{\pi}\int_0^\pi \left(z+\sqrt{z^2-1}\cos(t)\right)^\alpha dt.\qedhere\]
\end{proof}

We are now ready to prove Corollary~\ref{cauchyMechanism}.

\begingroup
\def\thecorollary{\ref{cauchyMechanism}}
\begin{corollary}[Cauchy Mechanism]
    Let $d \in \mathbb{N}^*$. We denote $Q_\alpha$ the Legendre polynomial of integer index $\alpha > 1$ and $\overline{Q_\alpha}$ as the polynomial derived from $Q_\alpha$ by retaining only its non-negative coefficients. Let $k \geq 2$ and $q \geq 1$ such that $kq(\alpha-1)/2$ is an integer. We note: \[\Delta^{dkq(\alpha-1)}_G = \underset{\substack{(s_i,s_j) \in S\\ \theta_i\in \Theta}}{\max}W_{dkq(\alpha-1)}\left( P(f(X)|s_i,\theta),P(f(X)|s_j,\theta)\right),\] with  $W_{dkq(\alpha-1)}$ computed with the $l_2$ norm. Then, $\mathcal{M}(X) = f(X) + V$ with $V = (V_1,\dots,V_d) \overset{iid}{\sim} \GCauchy\left(0,\lambda,k\right)$ is $\Bigg(\alpha,\frac{d\log\frac{\beta_{k,\lambda} \pi}{\lambda}\overline{Q}_{kq(\alpha-1)/2}\left(1 + \left(\frac{\Delta^{dkq(\alpha-1)}_G}{d\lambda}\right)^2\right)}{q(\alpha-1)}\Bigg)$-RPP.
\end{corollary}
\addtocounter{corollary}{-1}
\endgroup
\begin{proof}
We apply Theorem~\ref{EGWM} and compute $\Delta^{\zeta,q,\alpha}_G$ to prove our claim. We start by noticing that $h : z \mapsto \log Q_{k(\alpha-1)/2}\left(z\right)$ is convave for $z \geq 0$.
It is shown by factorizing the polynomial $Q_{k(\alpha-1)/2}$. It is known that all roots $r_1,\dots,r_{k(\alpha-1)/2}$ of Legendre polynomials are real, distinct from each other and lie in $(-1,1)$. Then, for $z \in \mathbb{R}$, $Q_{k(\alpha-1)/2}(z) = \Pi_{i=1}^{k(\alpha-1)/2} (z-r_i)$. \newpage Then, $h(z) = \log Q_{k(\alpha-1)/2}\left(z\right) = \sum_{i=1}^{k(\alpha-1)/2}\log \left(z-r_i\right)$, which is concave for $z\geq 0$.
Then, for $z_1,\dots,z_d \geq 0$, we have $\sum_{i=1}^d h(z_i) \leq d h(\sum_{i=1}^d z_i/d)$.
    We have:
        \begin{align*}
        \Delta^{\zeta,q,\alpha}_G &= \underset{\substack{(s_i,s_j) \in S\\ \theta_i\in \Theta}}{\max}\text{ }\underset{P(X,Y) \in \Gamma(\mu_i^\theta,\mu_j^\theta)}{\inf}\mathbb{E}\left[e^{q(\alpha-1)D_{q(\alpha-1)+1}(\zeta^{\otimes d} ,\zeta^{\otimes d} \ast (X-Y))}\right]\\
        &= \underset{\substack{(s_i,s_j) \in S\\ \theta_i\in \Theta}}{\max}\text{ }\underset{P(X,Y) \in \Gamma(\mu_i^\theta,\mu_j^\theta)}{\inf}\left[e^{q(\alpha-1)\sum_{i=1}^d D_{q(\alpha-1)+1}(\zeta\ast (Y_i-X_i) ,\zeta )}\right] \text{ independence of the noise}\\
        &\leq \underset{\substack{(s_i,s_j) \in S\\ \theta_i\in \Theta}}{\max}\text{ }\underset{P(X,Y) \in \Gamma(\mu_i^\theta,\mu_j^\theta)}{\inf}\left[e^{\sum_{i=1}^d \log \frac{\beta_{k,\lambda} \pi}{\lambda}Q_{kq(\alpha-1)/2}\left(1 + \frac{(X_i-Y_i)^2}{\lambda^2}\right)}\right] \text{ Lemma~\ref{ShiftCauchy}}\\
        &\leq \underset{\substack{(s_i,s_j) \in S\\ \theta_i\in \Theta}}{\max}\text{ }\underset{P(X,Y) \in \Gamma(\mu_i^\theta,\mu_j^\theta)}{\inf}\mathbb{E}\left[\left(\frac{\beta_{k,\lambda} \pi}{\lambda}Q_{kq(\alpha-1)/2}\left(1 + \frac{\|X-Y\|^2}{d^2\lambda^2}\right)\right)^d\right]\text{ Concavity inequality}\\
        &\leq \left(\frac{\beta_{k,\lambda} \pi}{\lambda}\right)^{d} \underset{\substack{(s_i,s_j) \in S\\ \theta_i\in \Theta}}{\max}\text{ }\underset{P(X,Y) \in \Gamma(\mu_i^\theta,\mu_j^\theta)}{\inf}\sum_{i=0}^{dkq(\alpha-1)/2} a_i\mathbb{E}\left[\left(1 + \frac{\|X-Y\|^2}{d^2\lambda^2}\right)^i\right]\substack{\text{ $\overline{Q}_{kq(\alpha-1)/2}^d$ is a polynomial}\\\text{ (degree $dkq(\alpha-1)/2$)}}\\
        &\leq  \left(\frac{\beta_{k,\lambda} \pi}{\lambda}\right)^{d}\underset{\substack{(s_i,s_j) \in S\\ \theta_i\in \Theta}}{\max}\text{ }\underset{P(X,Y) \in \Gamma(\mu_i^\theta,\mu_j^\theta)}{\inf}\sum_{i=0}^{dkq(\alpha-1)/2}\sum_{l=0}^{i}\binom{i}{l}a_i\mathbb{E}\left[ \frac{\|X-Y\|^{2i}}{\lambda^{2i}}\right]\\
        &\leq  \left(\frac{\beta_{k,\lambda} \pi}{\lambda}\right)^{d}\underset{\substack{(s_i,s_j) \in S\\ \theta_i\in \Theta}}{\max}\text{ }\underset{P(X,Y) \in \Gamma(\mu_i^\theta,\mu_j^\theta)}{\inf}\sum_{i=0}^{dkq(\alpha-1)/2}\sum_{l=0}^{i}\binom{i}{l}a_i\frac{\mathbb{E}\left[ \|X-Y\|^{dkq(\alpha-1)}\right]^\frac{2i}{dkq(\alpha-1)}}{d^{2i}\lambda^{2i}} \substack{\text{ Jensen inequality} \\ \text{ ($2i \leq dkq(\alpha-1)$)}}\\ 
        &\leq  \left(\frac{\beta_{k,\lambda} \pi}{\lambda}\right)^{d}\underset{\substack{(s_i,s_j) \in S\\ \theta_i\in \Theta}}{\max}\sum_{i=0}^{dkq(\alpha-1)/2}\sum_{l=0}^{i}\binom{i}{l} a_i\frac{W_{dkq(\alpha-1)}(\mu_i^\theta,\mu_j^\theta)^{2i}}{d^{2i}\lambda^{2i}} \text{ by definition of $W_{dkq(\alpha-1)}$}\\
        &\leq  \left(\frac{\beta_{k,\lambda} \pi}{\lambda} \underset{\substack{(s_i,s_j) \in S\\ \theta_i\in \Theta}}{\max}\overline{Q}_{kq(\alpha-1)/2}\left(1 + \frac{W_{dkq(\alpha-1)}(\mu_i^\theta,\mu_j^\theta)^2}{d^2\lambda^2}\right)\right)^{d}.\qedhere
    \end{align*}
\end{proof}

\subsection{Utility of the DAGWM (Proposition~\ref{DAGWMUtility})}
\label{appendixDAGWMUtility}

 In order to prove analyze the utility of the DAWGM, we resort to the following lemma.

\begin{lemma}
\label{DAGWMUtilLemma}
    Let $\alpha > 1$. Let $\zeta$ be a distribution of $\mathbb{R}^d$ and $W$ a random variable of $\mathbb{R}^d$ such that $\|W\| \leq z \text{ a.s}$. For $N$ drawn from $\zeta$ and independent of $W$, we have: \[D_\alpha(N+W,N) \leq R_\alpha(\zeta,z).\]
\end{lemma}

\begin{proof}
    We have:
    \begin{align*}
        e^{(\alpha-1) D_\alpha(N+W,N)} &= \int \frac{P_{W+N}(x)^\alpha}{P_N(x)^{\alpha-1}}dx \\
        &= \int \frac{\mathbb{E}[P_{N}(x-W)]^\alpha}{P_N(x)^{\alpha-1}}dx \\
        &\leq \int \int_{\|u\|\leq z} \frac{P_{N}(x-u)^\alpha}{P_N(x)^{\alpha-1}}P_W(u)dx du \\
        &= \int_{\|u\|\leq z} e^{(\alpha-1) D_\alpha(N+u,N)}P_W(u)du
        \\
        & \leq \int_{\|u\|\leq z} \sup_{\|x\| \leq z} e^{(\alpha-1) D_\alpha(N+z,N)}P_W(u)du \\
        &\leq \sup_{\|x\| \leq z} e^{(\alpha-1) D_\alpha(N+W,N)}.\qedhere
    \end{align*}
\end{proof}
Below, we make the informal result of Proposition~\ref{DAGWMUtility} precise and provide its proof.

\begin{proposition}[Utility of the DAGWM]
    Let $(\mathcal{S},\mathcal{Q},\Theta)$ be a Pufferfish framework, and let $\mathcal{M}(X) = f(X) + N$, where $X \sim \theta \in \Theta$, $N \sim \zeta \in \mathcal{P}(\mathbb{R}^d)$ and $f$ is a numerical query. Let $\alpha > 1$. Then, $R_\alpha(\zeta, \Delta_G)$ as defined in Theorem~\ref{GWM} is greater or equal to $\frac{\log(\Delta_{G}^{\zeta,1,\alpha})}{\alpha-1}$ defined in Theorem~\ref{EGWM}.
\end{proposition}

\begin{proof}
    By definition: for $(s_i,s_j) \in \mathcal{Q}, \theta \in \Theta$, we note $(\mu_i^\theta,\mu_j^\theta)=  (P(f(X)|s_i,\theta),P(f(X|s_j,\theta)))$, and if $(f(X)_{|s_i,\theta},f(X)_{|s_j,\theta})~\sim \pi^* \in \Gamma(\mu_i^\theta,\mu_j^\theta)$ realises the optimal transport plan for $W_\infty(\mu_i^\theta,\mu_j^\theta)$:
    \[\|f(X)_{|s_i,\theta}-f(X)_{|s_j,\theta}\| \leq W_\infty(\mu_i^\theta,\mu_j^\theta) \text{  a.s.}\]

    Using Lemma~\ref{DAGWMUtilLemma}, We have:
    \begin{align*}
        \mathbb{E}\left[e^{(\alpha-1)D_{\alpha}(\zeta,\zeta\ast(f(X)_{|s_i,\theta}-f(X)_{|s_j,\theta}))}\right] &= \mathbb{E}\left[e^{(\alpha-1)D_{\alpha}(\zeta\ast(f(X)_{|s_j,\theta}-f(X)_{|s_i,\theta}),\zeta)}\right]\\ &\leq e^{(\alpha-1)R_\alpha(\zeta,W_\infty(\mu_i^\theta,\mu_j^\theta))}.
    \end{align*}
     It follows:
    \begin{align*}
    \Delta^{\zeta,1,\alpha}_G &= \underset{\substack{(s_i,s_j) \in S\\ \theta_i\in \Theta}}{\max}\text{ }\underset{P(X,Y) \in \Gamma(\mu_i^\theta,\mu_j^\theta)}{\inf}\mathbb{E}\left[e^{(\alpha-1)D_{\alpha}(\zeta ,\zeta\ast (X-Y ))}\right]\\ 
    &\leq \underset{\substack{(s_i,s_j) \in S\\ \theta_i\in \Theta}}{\max} \mathbb{E}\left[e^{(\alpha-1)D_{\alpha}(\zeta,\zeta\ast(f(X)_{|s_i,\theta}-f(X)_{|s_j,\theta}))}\right]\\
    &\leq e^{(\alpha-1)R_\alpha(\zeta,W_\infty(\mu_i^\theta,\mu_j^\theta))}.
    \end{align*}
    Finally : \[\frac{\log(\Delta^{\zeta,1,\alpha}_G)}{\alpha-1} \leq \underset{\substack{(s_i,s_j) \in S\\ \theta_i\in \Theta}}{\max}R_\alpha(\zeta,W_\infty(\mu_i^\theta,\mu_j^\theta)) = R_\alpha(\zeta,\Delta_G).\qedhere\]
\end{proof}

\section{Privacy Amplification by Iteration (Section~\ref{compPABI})}
\label{PABIAppendix}
\subsection{Parallel Composition}
 Assessing the privacy guarantees of composition in RPP may be challenging. As a matter of fact, there does not exist, to our knowledge, any theorem stating the mechanism-agnostic privacy guarantees of sequential composition in Pufferfish privacy. However, we can recover a staightforward result of parallel composition for the RPP framework.

\begin{proposition}[RPP parallel composition for queries performed over independent datasets]
\label{parallelComp}
    Let $m > 0$ and  $(\mathcal{S},\mathcal{Q},\Theta_k)$ be Pufferfish frameworks corresponding to each dataset $X_k \sim P(\cdot |s_i^k,\theta_k)$. We assume that each secret $s_i^k$ is independent of the distributions $\theta_l$, for $l \neq k$ and that $\mathcal{Q}$ only contains pairs of the form $(s_i^k,s_j^k)$. For all $ k \in \{1,\dots,n\}$, let $\mathcal{M}_k(X_k)$ be mechanisms that satisfy $(\alpha,\varepsilon_k)$-RPP. Let $\Theta = \left\{\bigotimes_{k=1}^m \theta_k ; \forall k \in \{1,\dots,m\}, \theta_k \in \Theta_k\right\}$. Then, the mechanism $(\mathcal{M}_1,\dots,\mathcal{M}_m)$ satisfies $(\alpha,\max_k \varepsilon_k)$-RPP for $(\mathcal{S},\mathcal{Q},\Theta)$.
\end{proposition}
\begin{proof}
    Let $s_i^l,s_j^l \in \mathcal{Q}, \theta = \bigotimes_{k=1}^m \theta_k \in \Theta$. 
    \begin{align*}
        D_\alpha(P(\mathcal{M}(X)|s_i^l,\theta),P(\mathcal{M}(X)|s_j^l,\theta)) &= D_\alpha\big(P\left((\mathcal{M}_1(X_1),\dots,\mathcal{M}_n(X_n))|s_i^l,\otimes_{k=1}^m \theta_k\right),\\&P\left((\mathcal{M}_1(X_1),\dots,\mathcal{M}_n(X_n)|s_j^l,\otimes_{k=1}^m \theta_k\right)\big)\\
        &= \sum_{k=1}^n D_\alpha\left(P\left(\mathcal{M}_k(X_k)|s_i^l, \theta_l\right),P\left(\mathcal{M}_k(X_k)|s_j^l,\theta_k\right)\right)\\
        &= D_\alpha\left(P\left(\mathcal{M}_l(X_l)|s_i^l, \theta_l\right),P\left(\mathcal{M}_l(X_l)|s_j^l,\theta_l\right)\right) \leq \varepsilon_l.\qedhere
    \end{align*}
\end{proof}

This theorem states that if an adversary assumes that the dataset can be split into independent parts and if the secrets have some form of separability, such as in our Example~\ref{example}, it is possible to apply a different RPP mechanisms to each independent part while paying only for the maximum privacy loss, similar to the parallel composition result for differential privacy.

\subsection{Proof of Lemma~\ref{generalizedContraction}}
\begingroup
\def\thelemma{\ref{generalizedContraction}}
\begin{lemma}[Dataset Dependent Contraction lemma]

    Let $\psi$ be a contractive map in its first argument on $(\mathcal{Z},\|\cdot\|)$.  Let $X, X'$ be two r.v's. Suppose that $\sup_{w}W_\infty(\psi(w,X),\psi(w,X'))\leq s$. Then, for $z > 0$:
    \[D^{(z+s)}_{\alpha}(\psi(W,X),\psi(W',X')) \leq D^{(z)}_{\alpha}(W,W').\] 
\end{lemma}
\addtocounter{lemma}{-1}
\endgroup
\begin{proof}
    This proof is similar to the contraction lemma of~\cite{Feldman2018}. Let $s > 0$ such that $\sup_w W_\infty(\psi(W,X),\psi(W,X')) \leq s$, we have, for $Y$ a v.a. such that $D_\alpha^{(z)}(W,W') = D_\alpha(Y,W')$ and $W_\infty(W,Y) \leq z$:
    \begin{align*}
        W_\infty(\psi(W,X),\psi(Y,X')) &\leq W_\infty(\psi(W,X),\psi(W,X')) + W_\infty(\psi(W,X'),\psi(Y,X')) \\
        &\leq s + W_\infty(W,Y)\\
        &\leq s+z.
    \end{align*}
    It follows that: \[D^{(z+s)}_{\alpha}(\psi(W,X),\psi(W',X')) \leq D_\alpha(\psi(Y,X'),\psi(W',X')) \leq  D_\alpha(Y,W') = D^{(z)}_{\alpha}(W,W').\qedhere\]
\end{proof}

\subsection{Proof of Theorem~\ref{newPABI}}
\begingroup
\def\thetheorem{\ref{newPABI}}
\begin{theorem}[Dataset Dependent PABI]
    Let $X_T$ and $X'_T$  denote the output of $\CNI_T(W_0,\{\psi_t\},\{\zeta_t\},X)$ and $\CNI_T(W_0,\{\psi_t\},\{\zeta_t\},X')$. Let $s_t =\sup_{w}W_\infty(\psi(w,X_t),\psi(w,X_t'))$. Let $a_1,\dots,a_T$ be a sequence of reals and let $z_t = \sum_{i\leq t}s_i-\sum_{i\leq t}a_i$. If $z_t \geq 0$ for all $t$, then, we have:
    \[D^{(z_T)}_{\alpha}(X_T,X'_T) \leq \sum_{t=1}^{T}R_\alpha(\zeta_t,a_t).\]
\end{theorem}
\addtocounter{theorem}{-1}
\endgroup
\begin{proof}
    The proof is similar to the original PABI proof of~\cite{Feldman2018}. It is obtained by induction by replacing in the original PABI proof $s_t = \sup_{w\in \mathbb{R}^d,x,x' \in \mathcal{X}}\|\psi(w,x)-\psi(w,x')\|$ by $s_t =\sup_{w}W_\infty(\psi(w,X_t),\psi(w,X_t'))$ and using the dataset dependent contraction lemma (Lemma~\ref{generalizedContraction}).
\end{proof}

\section{Applications (Section~\ref{AppsandExps})}
\label{AppsandExpsappendix}

\subsection{Proof of Proposition~\ref{weaklydependent}}
\begingroup
\def\theproposition{\ref{weaklydependent}}
\begin{proposition}
    Let $\lambda > 0$, $(\mathcal{S},\mathcal{Q},\Theta)$ a Pufferfish framework. We note $\Delta$ the sensitivity of a numerical mechanism $f$ and: 
    \begin{align*}
        \Theta_\lambda = \{&\theta \in P(\mathcal{X}^{n});\\& \sup_{s_i \in \mathcal{S}}W_\infty(P(f(X)|s_i,\theta),P(f(X)|s_i,\theta^\otimes))\leq \lambda\}.
    \end{align*} Then, if $\Theta \subseteq \Theta_\lambda$, $\Delta_G\leq 2\lambda + \Delta$. 
\end{proposition}
\addtocounter{proposition}{-1}
\endgroup
\begin{proof}
    It is a direct consequence of triangle inequality for the $W_\infty$ distance: for $(s_i,s_j)\in \mathcal{Q}$, \begin{align*}
        W_\infty(P(f(X)|s_i,\theta),P(f(X)|s_j,\theta)) &\leq W_\infty(P(f(X)|s_i,\theta),P(f(X)|s_j,\theta)) \\&+ W_\infty(P(f(X)|s_i,\theta^\otimes),P(f(X)|s_j,\theta^\otimes)) \\&+ W_\infty(P(f(X)|s_j,\theta^\otimes),P(f(X)|s_j,\theta))\\
        &\leq \Delta + \sup_{s_i \in \mathcal{S}}W_\infty(P(f(X)|s_i,\theta),P(f(X)|s_i,\theta^\otimes))
    \end{align*}
\end{proof}
\subsection{Attribute Inference Setting}
\label{AppendixAttribute}
\subsubsection{Definitions}
We recall the definition of Dataset Attribute Privacy from~\cite{Zhang2022}:
\begin{definition}[Dataset Attribute Privacy~\cite{Zhang2022}]
    Let $X = (X_i^1,\dots,X_i^m)$ a record with $m$ attributes that is sampled from an unknown distribution $\mathcal{D}$, and let $X = (X^1,\dots,X^m)$ be a dataset of $n$ records sampled i.i.d from $\mathcal{D}$, where $X^i$ denotes the (column) vector containing values of $i$th attribute of every record. Let $C \subseteq [m]$ be the set of indices of sensitive attributes, and for each $i \in C$, let $g_i(X^i)$ be a function with codomain $\mathcal{U}_i$. A mechanism $\mathcal{M}$ satisfies $(\varepsilon,\delta)$-dataset attribute privacy if it is $(\varepsilon,\delta)$-Pufferfish private for the following framework $(\mathcal{S}, \mathcal{Q}, \{\theta\})$:
    \begin{itemize}
        \item Set of secrets $\mathcal{S} = \{s_i^a \overset{def}{=} \boldsymbol{1}[g_i(X^i) \in \mathcal{U}_i^a \subseteq \mathcal{U}_i; i \in C]\}$.
        \item Set of secret pairs $\mathcal{Q} = \{(s_i^a,s_i^b) \in \mathcal{S}\times \mathcal{S}; i \in C\}$.
        \item $\Theta$ is a set of possible distributions $\theta$ over the dataset $X$. For each possible distribution $\mathcal{D}$ over records, there exists a $\theta_\mathcal{D} \in \Theta$  that corresponds to the distribution over $n$ i.i.d. samples from $\mathcal{D}$.
    \end{itemize}
\end{definition}
\label{AppendixAttributeDef}
\newpage
\subsubsection{Experiments}
\label{experiments}
We make some experiments in order to highlight strictly better utility than DP mechanisms for real world datasets. We compare the sensitivity of DP, noted $\Delta$, the sensitivity of the GWM, noted $\Delta_G$ and the sensitivity of the DAGWM for the Cauchy distribution, noted $\Delta_{G,2}$.
\subsubsection{Attribute inference setting}
\paragraph{Datasets}
The datasets are taken from the UCI Machine Learning Repository~\cite{UCI}.
\begin{itemize}
    \item Student grade prediction: the Student Performance dataset~\cite{Cortez2014} has been collected to predict student grades. We want to release the column $X$ of final math grades (0-20) while protecting the privacy of the values of the column $S$ representing attendance in extra paid classes. We find  Here, $W_\infty(P(X|S=\text{"no"}),P(X|S=\text{"yes"})) = 8$, $W_2(P(X|S=\text{"no"}),P(X|S=\text{"yes"})) \approx 2.76$. The distribution of $X$ conditioned on $S$ is shown in the accompanying figure.
    \begin{figure}[ht]
    \centering
    \includegraphics[width=0.4\textwidth]{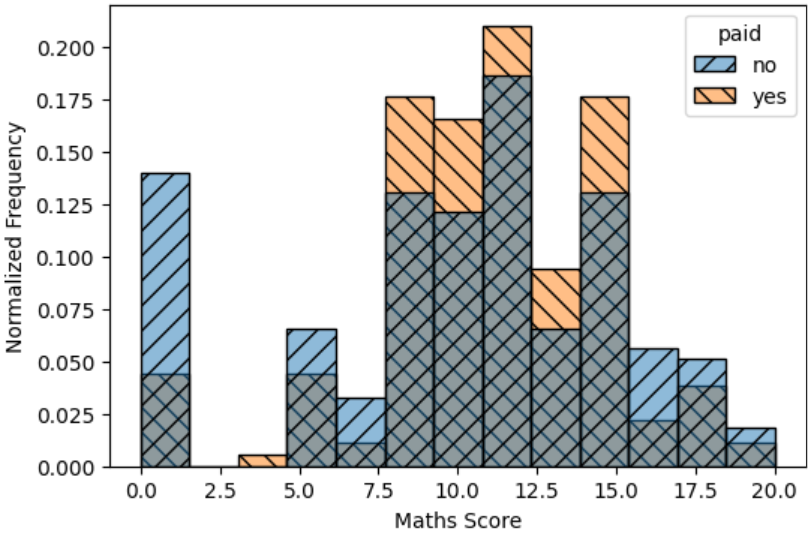}
    \caption{Distribution of the student scores on attendance in extra paid classes from the Student Performance dataset.}
    \end{figure}
 
    \item Heart disease prediction: the Heart Disease dataset~\cite{Janosi1988} has been collected to predict heart disease diagnosis. We want to release the column $X$ of ages (integer) while protecting the privacy of the values of the column $S$ representing heart disease diagnosis (represented by integer values in (0-4). $\max_{\text{disease}_i,\text{disease}_j}W_\infty(P(X|S=\text{disease}_i),P(X|S=\text{disease}_j)) = 8$, $\max_{\text{disease}_i,\text{disease}_j}W_2(P(X|S=\text{disease}_i),P(X|S=\text{disease}_j))\approx 7.80$. The distribution of $X$ conditioned on $S$ is shown in the accompanying figure.
    \begin{figure}[ht]
        \centering 
        \includegraphics[width=\textwidth]{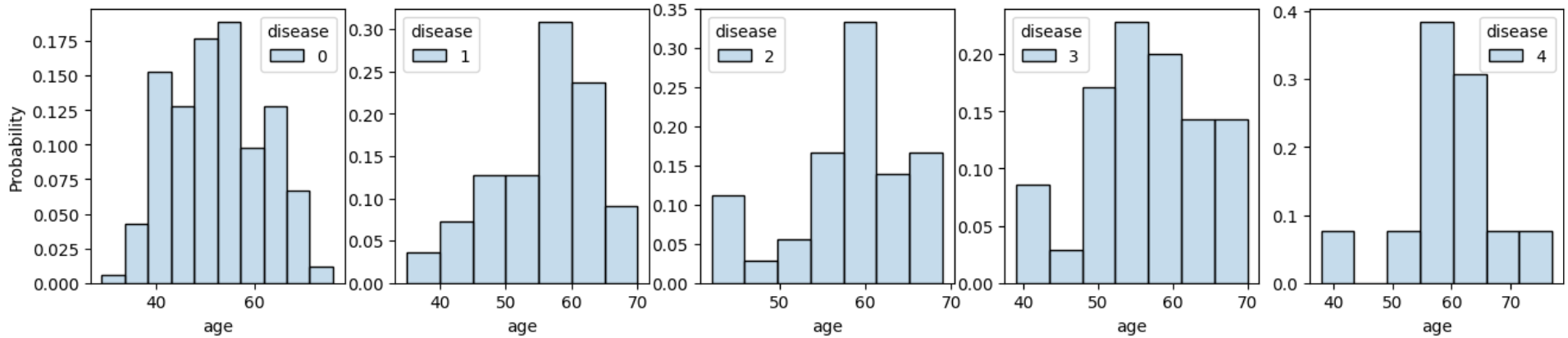}
        \caption{Distribution of ages conditioned on heart diagnosis from the Heart Disease dataset }
        \label{fig:enter-label}
    \end{figure}
    
    \item Salary prediction: the Adult dataset~\cite{Becker1996} is a popular dataset allowing to predict the salary of an individual. We want to release the column $X$ of salaries in $\{\leq 50K, >50K\}$ ($\{0,1\}$ for privacy analysis) while protecting the privacy of the values of the column $S$ representing the individual race. We find $\max_{\text{race}_i,\text{race}_j}W_\infty(P(X|S=\text{race}_i),P(X|S=\text{race}_j)) = 1$, $\max_{\text{race}_i,\text{race}_j}W_2(P(X|S=\text{race}_i),P(X|S=\text{race}_j))\approx 0.42$. The distribution of $X$ conditioned on $S$ is shown in the accompanying figure.
    \begin{figure}[ht]
    \centering
    \includegraphics[width=\textwidth]{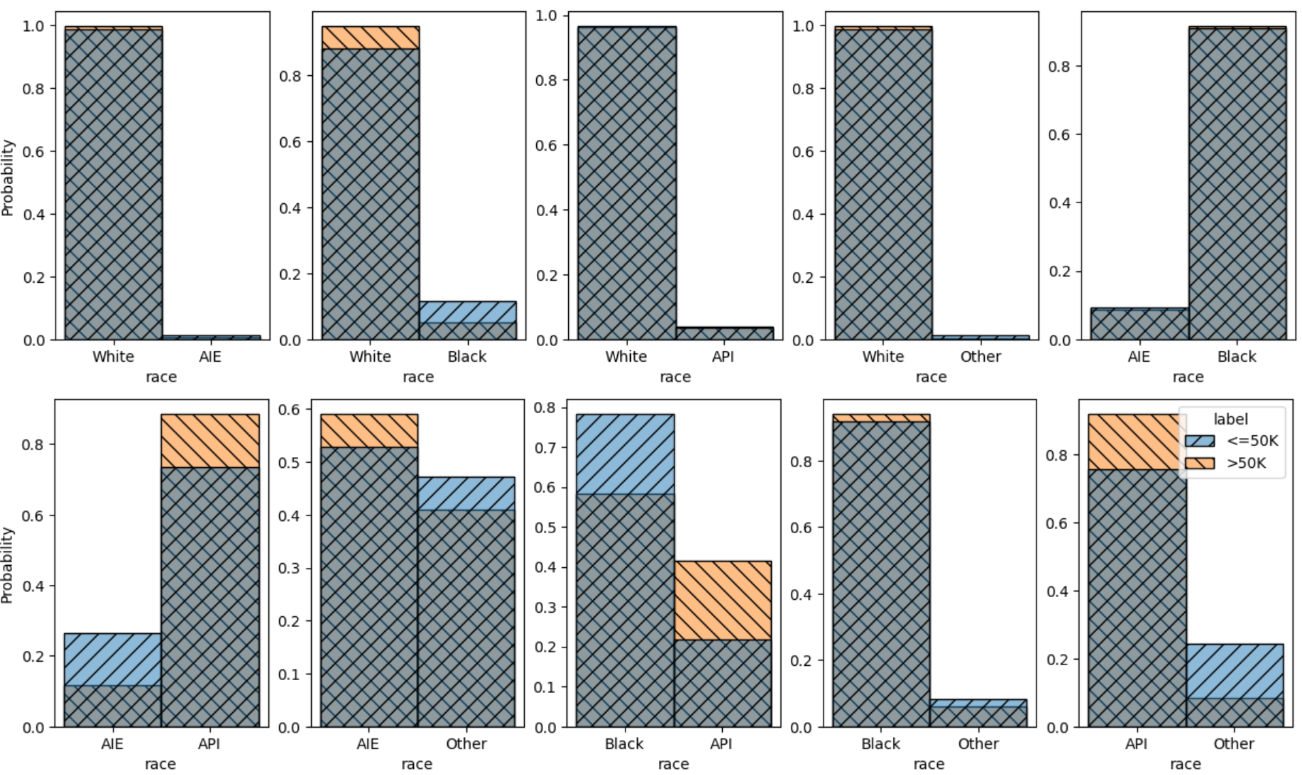}    
    \caption{Distribution of the salary label conditioned on every pair of races from the Adult dataset.}
    \end{figure}
\end{itemize}
\subsection{Attribute Inference for Gaussian Data}
\label{AppendixAttributeGaussian}
\begin{proposition}[Multiple attribute inference with Gaussian data]
    Let $M \in \mathbb{R}^{l_1 \times m}, N \in \mathbb{R}^{l_2 \times m}$, with $l_1,l_2 \leq m$. We assume that the adversary has a prior $\theta = \mathcal{N}(\mu,\Sigma) \in P(\mathbb{R}^m)$, with $X = (X_1^1,\dots,X_n^m)$. Then, considering the secrets $s_i^a =\{NX_i = a\}$, with $a \in K$ a compact of $\mathbb{R}^m$, the pairs of secrets $\mathcal{Q} = \{(s_i^a,s_i^b);i\in \Iintv{1,d},a,b \in K\}$ and the linear numerical query $f : x = (x_1,\dots,x_n) \mapsto (M x_1,\dots,Mx_n)$, \[\Delta_G \leq \underset{a,b \in K}{\max}\|\Cov(MX_1,NX_1)\Cov(NX_1)^{-1}(a-b)\|\] for the Pufferfish framework $(\mathcal{S},\mathcal{Q},\{\theta\})$.
\end{proposition}
\begin{proof}
    For $i \in \Iintv{1,n}$, we have: $\begin{pmatrix} M \\ N    
    \end{pmatrix}X_i \sim \mathcal{N}\left(\begin{pmatrix}
        M \mu \\ N \nu
    \end{pmatrix}, \begin{pmatrix}
        M\Sigma M^T & M\Sigma N^T \\ N\Sigma M^T & N \Sigma N^T
    \end{pmatrix}\right)$. 
    Then, \[MX_i | NX_i = a \sim \mathcal{N}\left(M\mu + M\Sigma N^T (N\Sigma N^T)^{-1}(a-N \mu),M\Sigma M^T - M\Sigma N^T (N \Sigma N^T)^{-1} N \Sigma M^T\right)\]
    We note $\Cov(MX_i,NX_i) = M \Sigma N^T$. Drawing $Y \sim P(MX_i | NX_i = a)$ and noting: \[Z = Y + \Cov(MX_i,NX_i)(N\Sigma N^T)^{-1}(b-a),\] we have $Z \sim P(MX_i | NX_i = b)$ and: \[\|Y-Z\| = \|\Cov(MX_i,NX_i)(N\Sigma N^T)^{-1}(b-a)\|.\]
\end{proof}

\subsection{Proof of Proposition~\ref{diffusionprotection}}
\begingroup
\def\theproposition{\ref{diffusionprotection}}
\begin{proposition}
    Let $V : \mathbb{R}^d\to\mathbb{R}$ such that $\nabla^2 V \succcurlyeq C I_d$. For $\theta_0 \in P(\mathbb{R}^d)$, we note $\theta_t$ the distribution of $X_t$, with $(X_t)_{t\geq0}$ solution of the stochastic differential equation: $dX_t = -\nabla V(X_t)dt+\sqrt{2}dB_t$, where $(B_t)_{t\geq0}$ is a brownian motion. We note $\theta_{t_1,\dots,t_n}$ the distribution generating $X = (X_{t_1},\dots,X_{t_n})$ from the distribution of $(X_t)_{t\geq0}$. We consider the secrets $s^a =\{X_0= a\}$, with $a \in K$ a compact of $\mathbb{R}^d$, the pairs of secrets $\mathcal{Q} = \{(s^a,s^b);a,b \in K\}$ Then, the GWM of any $L$-Lipschitz query $f$ performed on $X$ has a sensitivity for the $l_1$ norm:
    \[\Delta_G \leq L \text{Diam}(K) \sum_{i=1}^n \exp(-2Ct_i) \] for the Pufferfish framework $(\mathcal{S},\mathcal{Q},\{\theta_{t_1,\dots,t_n}\})$.
\end{proposition}
\addtocounter{proposition}{-1}
\endgroup
\begin{proof}
    The proof can be obtained via a synchronous coupling argument, which can for example be found in~\cite{Villani2008}. Let $a,b \in K$. Let $(B_t)_{t\geq 0}$ be a brownian motion and we define:
    \begin{align*}
        X_t &= a - \int \nabla V(X_s)ds + \sqrt{2}B_t,\\
        Y_t &= b - \int \nabla V(Y_s)ds + \sqrt{2}B_t,\\
    \end{align*}
    with the same realization of $(B_t)_{t\geq 0}$ for the two processes. Noting $\alpha_t = W_t-V_t$, we have:
$\frac{d\alpha_t}{dt} = -(\nabla R_\lambda(W_t)-\nabla R_\lambda(V_t))$ and, by convexity of V: 
\[\frac{d\|\alpha_t\|_1^2}{dt} = -2\langle \nabla R_\lambda(W_t)-\nabla R_\lambda(V_t), W_t-V_t\rangle \leq -2C\|\alpha_t\|_1.\]
Gronwall's lemma implies that $\forall t \geq 0, \|\alpha_t\|_1\leq e^{-Ct}\|a-b\|_1$. 
Then, \[\|f(X_{t_1},\dots,X_{t_n})- f(Y_{t_1},\dots,Y_{t_n})\|_1 \leq L \sum_{i=1}^n \|\alpha_{t_i}\|_1 \leq L \|a-b\|_1 \sum_{i=1}^n \exp(-2Ct_i)\]
    
\end{proof}
\newpage

\subsection{Application of PABI to Convex Optimization (Section~\ref{sectionPABIConvex})}
\subsubsection{Setup of convex optimization for PABI}
\label{AppendixPABISetup}
Here is the setup for projected noisy stochastic gradient descent in the convex setting:
        \begin{itemize}
        \item $f$ is $L$-Lipschitz in its first argument: there exists $L> 0$ such that $\forall x\in \mathcal{X},w_1,w_2 \in \mathbb{R}^d$, \[\|f(w_1,x)-f(w_2,x)\| \leq L\|w_1-w_2\|.\]
        \item $f$ is $\beta$-smooth in its first argument: there exists $\beta > 0$ such that $\forall x\in \mathcal{X},w_1,w_2 \in \mathbb{R}^d$, \[\|\nabla_wf(w_1,x)-\nabla_wf(w_2,x)\| \leq \beta\|w_1-w_2\|.\]
        \item $f$ satisfies the following condition: $\forall x_1,x_2\in \mathcal{X},w_1\in \mathbb{R}^d, \exists C_{w_1} > 0$ such as : \[\|\nabla_wf(w_1,x_1)-\nabla_wf(w_1,x_2)\| \leq C_{w_1}\|x_1-x_2\|.\]
    \end{itemize}
\subsubsection{Application to DP}
\label{DPapplication}
\begin{lemma}[Example: DP as a special case]
    In the case of DP, each distribution $\theta\in\Theta$ corresponds to a prior of independence between the elements of the dataset. Let $\beta,\eta,\sigma, L, T > 0, \alpha >1$ such that $\eta > 2/\beta$. We set the secrets $\mathcal{S} = \left\{s_i^a \overset{def}{=}\{X_i = a\}; a \in \mathcal{X}\right\}$ and the pairs of secrets : $\mathcal{Q} = \{(s_i^a,s_i^b); a,b \in \mathcal{X}\}$. Let $(X,X') \sim \pi \in \Gamma(P(X|s_i^a),P(X|s_i^b))$.  Let $f$ be an objective function which is convex, $\beta$-smooth and $L$-Lipschitz. Let $\mathcal{K} \subset \mathbb{R}^d$ be a compact set. Let $W_0 = W_0' \in \mathcal{K}$ be the original weight of the stochastic gradient descent and $\psi$ the update function of the projected noisy stochastic gradient descent of learning rate $\eta$. Let $\zeta = \mathcal{N}(0,\sigma^2\eta^2I_d)$ be the noising distribution. For $t \in \Iintv{0,T}$, we define $W_t = \CNI_t(W_0,\psi,\zeta,X), W_t = \CNI_t(W_0',\psi,\zeta,X')$. Then, Theorem~\ref{newPABI} allows to obtain: \[D^{(z_T)}_{\alpha}(X_T,X'_T) \leq \frac{2\alpha L^2}{\sigma^2(T-i+1)}.\]
    This recovers the results of \citet{Feldman2018} for the case of DP-SGD. 
\end{lemma}
\begin{proof}
    Let $\sigma > 0$. Let $(s_i^a,s_i^b) \in \mathcal{Q}, \theta \in \Theta$, with $\theta$ representing a prior of independence. Then, for $t\in \Iintv{1,T}$,  $(X,X') \sim \pi \in \Gamma(P(X|s_i^a),P(X|s_i^b))$, $s_t = \sup_wW_\infty(\psi(w,X_t),\psi(w,X'_t)) = \begin{cases}
        \sup_w\|\psi(w,a)-\psi(w,b)\| \text{ if $t=i$}\\
        0 \text{ else}
    \end{cases}$, $\zeta_t = \mathcal{N}(0,(\eta\sigma)^2I_d)$. Then, setting $a_t = \begin{cases}
        \frac{s_i}{T-i+1} \text{ if $t \geq i$}\\
        0\text{ else}
    \end{cases}$, we get:
    \begin{align*}
        D^{(z_T)}_{\alpha}(X_T,X'_T) &\leq \sum_{t=i}^{T}R_\alpha\left(\zeta_t,\frac{\sup_w\|\psi(w,a)-\psi(w,b)\|}{T-i+1}\right)\\
        &\leq \sum_{t=i}^{T}\frac{\alpha \sup_w\|\psi(w,a)-\psi(w,b)\|}{2\eta^2\sigma^2(T-i+1)^2}\\
        &\leq \frac{2\alpha L^2}{\sigma^2(T-i+1)},
    \end{align*}
which is the bound of Theorem 23 of~\citet{Feldman2018}.
\end{proof}
\subsubsection{PABI bounds for Gaussian datasets}
\label{appendixPABIGaussian}
\begingroup
\def\theproposition{\ref{PABIGaussian}}
\begin{proposition}
    Assume that the adversary has a Gaussian prior $\theta$. Then,
    \begin{align*}
        &D_{\alpha}(W_T,W'_T) \leq \frac{\alpha\eta^2}{2\sigma^2 }\min(2L,\sup_{v\in \mathcal{K}}C_v\|a-b\|)^2\\&+\frac{\alpha\eta^2}{2\sigma^2 }\sum_{t\neq i}^{T}\min(2L,\sup_{v\in \mathcal{K}}C_v \|\Cov(X_t,X_i)\Cov(X_i)^{-1}(a-b)\|)^2.
    \end{align*}
\end{proposition}
\addtocounter{proposition}{-1}
\endgroup
\begin{proof}
    Let $t \in \Iintv{1,T}$ such that $t\neq i$. We want to find a upper bound to $W_\infty P(X_t | X_i=a),P(X_t | X_i=b))$. We note $X = (X_1^1,\dots,X_1^d,\dots, X_T^d) \sim \mathcal{N}(\mu,\Sigma)$, $X_t = (X_t^1,\dots,X_t^d) \sim \mathcal{N}(\mu_t,\Sigma_t)$ and $M_t^i = \begin{pmatrix}
        0_{(i-1)d} & I_d & 0_{(T-i)d} \\ 0_{(t-1)d} & I_d & 0_{(T-t)d}
    \end{pmatrix}$. Then, $\begin{pmatrix}
        X_i \\ X_t
    \end{pmatrix} = M_t^i X \sim \mathcal{N}\left(\begin{pmatrix}
        \mu_i \\ \mu_t
    \end{pmatrix}, \begin{pmatrix}
        \Sigma_{i} & \Sigma_{it} \\ \Sigma_{ti} & \Sigma_t
    \end{pmatrix}\right)$ ,  and:
    \[X_t|X_i=a \sim \mathcal{N}(\mu_t + 
     \Sigma_{ti}\Sigma_i^{-1}(a-\mu_i), \Sigma_t - \Sigma_{ti}\Sigma_i^{-1}\Sigma_{it}).\]
     Then, for $Y \sim X_t|X_i=a$ and $Z = Y + \Cov(X_t,X_i)\Cov(X_i)^{-1}(b-a)$, $Z \sim X_t|X_i=b$. For $t=i$, we have $W_\infty P(X_i | X_i=a),P(X_i | X_i=b)) = \|b-a\|$.
\end{proof}

\subsubsection{PABI bounds for decreasing dependencies}
\label{appendixPABIdecreasing}
The bounds of Section~\ref{sectionPABIConvex} can be improved in the case where $(W_\infty(X_t,X_t'))_t$ is non-increasing. 
\begin{proposition}
\label{improvedPABI}
    Taking the notations from Theorem~\ref{newPABI}, let $(X_t)$ and $(X'_t)$ be CNIs. Assume that $(s_t)_t = (W_\infty(X_t,X_t'))_t$ is non-increasing. Then, 
    \[D^{(z_T)}_{\alpha}(X_T,X'_T) \leq \sum_{t=1}^T R_\alpha\left(\zeta_t,\frac{\sum_{k=1}^{T} W_\infty(X_k,X_k')}{T}\right).\]
    When $\zeta_t = \zeta$ for all $t \in \{1,\dots,T\}$, the bound becomes:
    \[D^{(z_T)}_{\alpha}(X_T,X'_T) \leq T R_\alpha\left(\zeta,\frac{\sum_{t=1}^{T} W_\infty(X_t,X_t')}{T}\right).\]
\end{proposition}
\begin{proof}
    We assume that the sequence $(s_t)_t = (W_\infty(X_t,X_i))_t$ is decreasing.
    We apply Theorem~\ref{newPABI} to get new bounds. Compared to the analysis of Section~\ref{sectionPABIConvex}, it gives For $t \in \{1,\dots,T\}$, we take $a_t = \frac{\sum_{k=1}^T W_\infty(X_k,X_k')}{T}$, we have:\begin{align*}
    z_t = \sum_{k\leq t}s_i-\sum_{k\leq t}a_i &= \sum_{i\leq t}W_\infty(X_i,X_i')-\frac{1}{T}\sum_{i\leq t}\sum_{k=1}^T W_\infty(X_k,X_k')\\
    &= \frac{T-t}{T}\sum_{i\leq t}W_\infty(X_i,X_i') -\frac{t}{T}\sum_{t < i \leq T} W_\infty(X_i,X_i')\\
    &\geq \frac{T-t}{T}\left(\sum_{i\leq t}W_\infty(X_i,X_i') -t W_\infty(X_t,X_t')\right) \geq 0. \text{ $(s_t)$ is non-increasing} \qedhere
    \end{align*}
\end{proof}
We further analyze this new bound for the case of Gaussian noise : $\zeta_t = \zeta = \mathcal{N}(0,\sigma^2 I_d)$ for all $t$. Then, we have: \begin{equation}
\label{PABIimprovedEquation}
    D_{\alpha}(X_T,X'_T) \leq \frac{\alpha(\sum_{t=1}^{T} W_\infty(X_t,X_t'))^2}{2T\sigma^2}.
\end{equation}
We can compare this bound with the PABI bound of Section~\ref{sectionPABIConvex}:
\begin{equation}
\label{PABIoriginalEquation}
    D_{\alpha}(X_T,X'_T) \leq \frac{\alpha\sum_{t=1}^{T} W_\infty(X_t,X_t')^2}{2\sigma^2}.
\end{equation}
While the latter result (\ref{PABIoriginalEquation}) allows to derive privacy guarantees for composition in the Pufferfish framework, it does not ensure that privacy loss tends to $0$ as $T \to +\infty$ even when $\sum_{t=1}^{T} W_\infty(X_t,X_t')^2$ converges. However, when dependencies are decreasing over time, the privacy loss analysis is improved with (\ref{PABIimprovedEquation}).

We now compare our PABI bounds with the DP and the Group DP (which represent two extreme cases of our analysis). In order to do this, we illustrate the privacy loss as a function of the number of iterations.
 We let the secrets $s_a = \{X_1 = a\}, s_b = \{X_1 = b\}$ and for simplicity and visualization,  we stay in the Gaussian setting of Proposition~\ref{PABIGaussian}. We assume that each $X_t$ has a covariance matrix $\Cov(X_t) = I_d$ and the covariance between $X_t$ and $X_1$ is $\Cov(X_1,X_t) = \rho_t I_d$. This corresponds to the case where each the columns of the dataset are independent of each other but dependencies within each column are controlled by the parameter $\rho_t$. $\rho_t \overset{t \to + \infty}{\to} 0$ means that $X_t$ becomes increasingly independent of $X_1$ as $t$ increases, indicating that $X_t$ is less correlated with $X_1$ when they are far apart in the dataset. Using Proposition~\ref{improvedPABI} and Proposition~\ref{PABIGaussian}, we have, for $t \in \{1,\dots,T\}$ and $\rho_1 = 1$: $\|\Cov(X_t,X_i)\Cov(X_i)^{-1}(a-b)\| =  |\rho_t|\|a-b\|$. Then:
 \[D_{\alpha}(X_T,X'_T) \leq  \frac{\alpha\eta^2}{2T\sigma^2 }\left(\sum_{t = 1}^{T}\min(2L,\sup_{v\in \mathcal{K}}C_v |\rho_t|\|a-b\|)\right)^2 \leq \frac{\alpha\eta^2\|a-b\|^2 (\sup_{v\in \mathcal{K}}C_v)^2}{2T\sigma^2 }\left(\sum_{t = 1}^{T} |\rho_t|\right)^2.\]
 We set the following parameters for visualization: $L=\sigma=\eta=\sup C_v= \|a-b\|=1, \alpha=2$. Recall that standard DP corresponds to the absence of correlation ($\rho_t=0$), while Group DP corresponds to maximal correlation ($\rho_t=1$). In Figure~\ref{fig:rho-cst}, we show how our PABI bounds compare with the DP setting in the case where all elements of the dataset are equally correlated ($\rho_t=\rho$), highlighting the privacy gains over Group DP. In Figure~\ref{fig:rho-dec}, we show the convergence of the privacy loss to $0$ when the correlations vanish ($\rho_t \overset{t \to + \infty}{\to} 0$).

\begin{figure}[t]
    \centering
    \begin{minipage}[t]{0.48\textwidth}
        \centering
        \includegraphics[width=\textwidth]{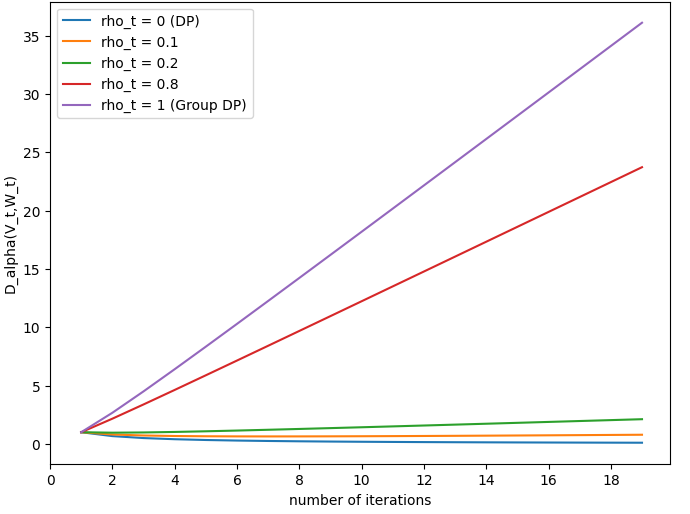} 
        \caption{Privacy loss as a function of the number of iterations for the following values of $\rho_t$: $0 \text{ (DP)}, 0.1,0.2,0.8 \text{ and } 1 \text{ (Group DP)}$.}
        \label{fig:rho-cst}
    \end{minipage}\hfill
    \begin{minipage}[t]{0.48\textwidth}
        \centering
        \includegraphics[width=\textwidth]{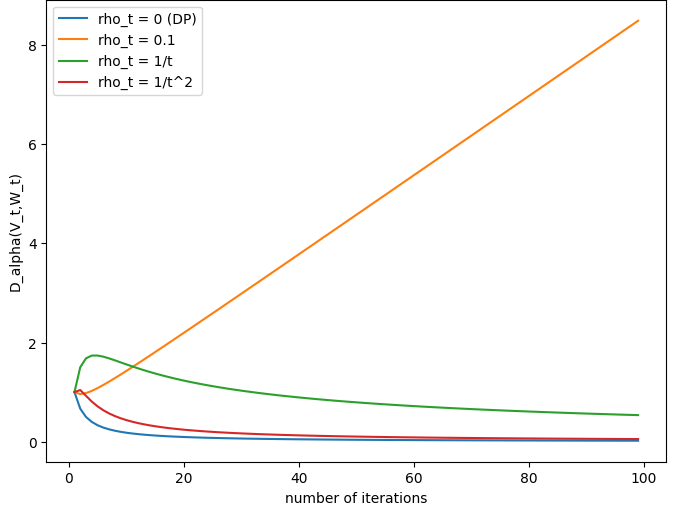} 
        \caption{Privacy loss as a function of the number of iterations for the following values of $\rho_t$ : $0 \text{ (DP)},  0.1,1/t \text{ and }1/t^2$.}
        \label{fig:rho-dec}
    \end{minipage}
\end{figure}

\end{document}